\documentclass{iopjournal}
\usepackage{amsmath}
\usepackage{amssymb}
\usepackage{amsthm}
\usepackage{graphicx}
\usepackage{subcaption}
\usepackage{ragged2e}

\newtheorem{theorem}{Theorem}
\newtheorem{proposition}{Proposition}

\usepackage[acronym]{glossaries}
\newacronym{flrw}{FLRW}{Friedmann-Lemaitre-Robertson-Walker}
\newacronym{gr}{GR}{General Relativity}
\newacronym{tegr}{TEGR}{Teleparallel Equivalent of General Relativity}
\newacronym{stegr}{STEGR}{Symmetric Teleparallel Equivalent of General Relativity}
\newacronym{lcdm}{$\Lambda$CDM}{Cold Dark Matter and Cosmological Constant}
\newacronym{sds}{SdS}{Schwarzschild-de Sitter}
\newacronym{tg}{TG}{Teleparallel Gravity}
\newacronym{stg}{STG}{Symmetric Teleparallel Gravity}

\begin{document}
\articletype{Paper} 

\title{Cosmologically Viable Solutions in Geometric Modified Gravity}

\author{P. A. G. Monteiro$^{1,2}$\orcid{0009-0005-1332-3807} and C. J. A. P. Martins$^{1,3,*}$\orcid{0000-0002-4886-9261}}

\affil{$^1$Centro de Astrof\'{\i}sica da Universidade do Porto, Rua das Estrelas, 4150-762 Porto, Portugal}

\affil{$^2$Faculdade de Ci\^encias, Universidade do Porto, Rua do Campo Alegre, 4150-007 Porto, Portugal}

\affil{$^3$Instituto de Astrof\'{\i}sica e Ci\^encias do Espa\c co, Universidade do Porto, Rua das Estrelas, 4150-762 Porto, Portugal}

\affil{$^*$Author to whom any correspondence should be addressed.}
\email{Carlos.Martins@astro.up.pt}

\keywords{Cosmology, Modified Gravity, Torsion, Non-metricity, Observational Tests}

\justifying

\begin{abstract}
\justifying
The discovery of the accelerated expansion of the universe highlighted General Relativity’s inability to naturally account for dark energy without invoking a finely tuned cosmological constant. In response, a wide range of alternative paradigms have been proposed. Among these, Teleparallel Gravity and Symmetric Teleparallel Gravity, which depart from the Riemannian framework of General Relativity and instead rely on torsion or non-metricity to describe gravitational interactions, have gained increasing attention in recent years. We explore extensions of these non-Riemannian approaches, aiming to replicate the observed late-time acceleration of the universe by emulating the cosmological constant's role. We also evaluate the consistency of these theories with local gravity constraints by studying their static, spherically symmetric solutions. We show that although some models can reproduce the desired cosmological behavior, they often fail to meet Solar System observational bounds, particularly through deviations in the predicted Eddington parameter. Our findings underscore the need for a unified approach that tests modified gravity theories across both cosmological and local scales.
\end{abstract}

\section{\label{sec1}Introduction}

The discovery of the accelerated expansion of the universe has triggered a flourishing interest in modified theories of gravity that aim to explain cosmic acceleration without invoking dark energy. This quest is informed by the need to reproduce the observationally successful cosmological behavior of the standard $\Lambda$CDM model. Among the possible alternatives, a compelling class of models extends the General Relativity (GR) framework by considering non-Riemannian geometries. In particular, gravity can be formulated not as a manifestation of spacetime curvature - as is the case in GR - but instead as mediated by other geometric quantities: torsion in Teleparallel Gravity (TG) and non-metricity in Symmetric Teleparallel Gravity (STG). These three formulations constitute the so-called geometric trinity of gravity \cite{beltran2019trinity}. Despite their distinct geometric constructions, these theories are dynamically equivalent at the level of the field equations.

To break this equivalence, attention has shifted to non-linear generalizations, specifically the theories of the form $f(R)$, $f(\mathbb{T})$ and $f(\mathbb{Q})$ where $R$, $\mathbb{T}$ and $\mathbb{Q}$ are the curvature, torsion, and non-metricity scalars, respectively. These extensions are no longer equivalent but can still yield similar field equations under specific conditions, maintaining an element of degeneracy \cite{wu2024correspondence}.

A valuable approach to investigating the cosmological viability of these theories is the reconstruction method, which seeks to determine the form of the gravitational Lagrangian that yields a desired background evolution, such as that of the $\Lambda$CDM model, from a given Hubble function $H(t)$. However, it has been shown that none of the single-variable theories $f(R)$, $f(\mathbb{T})$ and $f(\mathbb{Q})$ can reproduce the $\Lambda$CDM expansion without including a cosmological constant in the action \cite{LCDM_Reconstruction,gadbail2025_thesis}. This limitation has motivated the study of more general extensions, such as non-minimal matter-geometry couplings \cite{Gadbail_2023_Q_T,bertolami2007extra,harko2014nonminimal,das2024nonminimal} and boundary-term extended theories $f(\mathbb{T},B)$ and $f(\mathbb{Q},C)$ \cite{capozziello2025extended,Gadbail_2023_Q_C,Caruana_2020}. {We will analyze the reconstruction of boundary term extensions in order to formulate theories that reproduce the $\Lambda$CDM dynamics without invoking a cosmological constant.}

Importantly, any viable theory must not only explain cosmic acceleration but also pass stringent Solar System tests, which agree with General Relativity with high precision. The correct procedure to analyze these theories in the Solar System context involves matching the spherically symmetric exterior vacuum solution to an interior solution representing the Sun's matter distribution. This method was first applied in the analysis of the model $f(R)=R-\frac{\mu^4}{R}$ in \cite{erickcek2006} and later generalized to a broader class of $f(R)$ models \cite{chiba2007solar}. {We will apply these methods to teleparallel theories of gravity, demonstrating that results analogous to those of $f(R)$ gravity hold in this framework, in particular allowing us to rule out a specific subclass of such theories.}

In this paper, we aim to obtain cosmologically viable teleparallel theories via reconstruction and further constrain them using Solar System tests. In Section \ref{sec2}, we review the necessary concepts of non-Riemannian geometry that underpin these theories. In Section \ref{sec3}, we present the geometric trinity of gravity and its extensions. In Section \ref{sec4}, we review the reconstruction of $\Lambda$CDM in $f(R)$, $f(\mathbb{T})$ and $f(\mathbb{Q})$, showing that all require a cosmological constant. We then extend the reconstruction method to $f(\mathbb{T},B)$ and $f(\mathbb{Q},C)$, identifying a class of models that can reproduce $\Lambda$CDM without introducing a cosmological constant. In Section \ref{chapter solar}, we study these reconstructed theories in the context of Solar System tests, determining whether they satisfy or fail observational constraints. Moreover, we extend our analysis to a broader class of $f(\mathbb{T},B)$ and $f(\mathbb{Q},C)$ theories, generalizing the result obtained by Chiba et al. \cite{chiba2007solar} for $f(R)$ gravity. Finally, our results are summarized in Section \ref{sec6}. A more extended discussion of these results can be found in \cite{Thesis}.

\section{\label{sec2}Mathematical Formulation}

We start with a review of the relevant concepts, definitions and mathematical formalism, for the purpose of making the present work self-contained. A metric-affine geometry or Non-Riemannian geometry is defined by a triple $(\mathcal{M},g,\Gamma )$, where $\mathcal{M}$ is a 4-dimensional smooth manifold, $g$ is a Lorentzian metric on $\mathcal{M}$ (a smooth, non-degenerate, symmetric tensor field of type $(0,2)$ with signature $(-,+,+,+)$) and $\Gamma$ is a general affine connection. The metric allows the definition of norm of a vector: 
\begin{equation}\label{norma}
    ||V||=g_{\mu\nu}V^{\mu}V^{\nu}
\end{equation}
The affine connection defines the covariant derivative of a vector field $V^\nu$ as:
\begin{equation}\label{derivada covariante}
\nabla_{\mu}V^{\nu}=\partial_{\mu}V^{\nu}+\Gamma_{\hspace{0.5em}\mu\alpha}^{\nu}V^{\alpha}    
\end{equation}
Given a curve $x^{\mu}(\lambda)$, we define the direction derivative along the curve as
\begin{equation}
    \frac{D}{d\lambda}\equiv\frac{dx^\mu}{d\lambda}\nabla_\mu
\end{equation}
A vector $V^\mu$ is said to be parallel transported along the curve $x^{\mu}(\lambda)$ if 
\begin{equation}\label{Transporte Paralelo}
    \frac{D}{d\lambda}V^\mu=0
\end{equation} 

In a metric-affine geometry we can define the geometric quantities:

\textit{Riemann Tensor} 
\begin{equation}\label{Riemann Tensor}
R_{\hspace{0.4em}\mu\nu\rho}^{\alpha}=\partial_{\nu}\Gamma_{\hspace{0.4em}\rho\mu}^{\alpha}-\partial_{\rho}\Gamma_{\hspace{0.4em}\nu\mu}^{\alpha}+\Gamma_{\hspace{0.4em}\nu\lambda}^{\alpha}\Gamma_{\hspace{0.4em}\rho\mu}^{\lambda}-\Gamma_{\hspace{0.4em}\rho\lambda}^{\alpha}\Gamma_{\hspace{0.4em}\nu\mu}^{\lambda}
\end{equation}
\\
\textit{Torsion Tensor} 
\begin{equation}\label{Torsion Tensor}
T_{\hspace{1em}\mu\nu}^{\alpha}=\Gamma_{\hspace{1em}\mu\nu}^{\alpha}-\Gamma_{\hspace{1em}\nu\mu}^{\alpha}
\end{equation}
\\
\textit{Non-Metricity Tensor}
\begin{equation}\label{Non Metricity Tensor}
    Q_{\alpha\mu\nu}=\nabla_{\alpha}g_{\mu\nu}
\end{equation}
From the Riemann tensor, we construct the Ricci Tensor $R_{\mu\nu}$ and the Ricci Scalar $R$ as:
\begin{subequations}
\begin{align}
    R_{\mu\nu}&=R_{\hspace{0.4em}\mu\lambda\nu}^{\lambda}\\ 
R&=R_{\mu\nu}g^{\mu\nu} 
\end{align}
\end{subequations}
Contracting indices of the torsion and non-metricity tensors, we define the Torsion Vector $T_{\mu}$ and the First and Second Non-Metricity Tensor $Q_{\alpha}$ and $\bar{Q}_{\alpha}$:
\begin{subequations}
\begin{align}
T_{\mu}&=T_{\hspace{0.4em}\mu\alpha}^{\alpha} \\ 
Q_{\alpha}&=g^{\mu\nu}Q_{\alpha\mu\nu} \\
\bar{Q}_{\alpha}&=g^{\mu\nu}Q_{\mu\nu\alpha}
\end{align}
\end{subequations}

\subsection{Geometric Interpretation of Curvature, Torsion and Non-Metricity}

Consider a vector $V^\mu$ and parallel transport it around an infinitesimal parallelogram spanned by two vectors $A^\mu$ and $B^\mu$: first parallel transport $V^\mu$ in the direction of $A^\mu$, then along $B^\mu$ and then backward along $A^\mu$ and $B^\mu$. The change in the vector $V^\mu$ after being parallel transported around this closed loop is given by
\begin{equation}
    \delta{V^\mu}=R^\mu_{\hspace{0.5em}\nu \alpha \beta}V^\nu A^\alpha B^\beta 
\end{equation}
Thus, the curvature tensor measures the change of a vector when it is parallel transported along a close loop.

In Euclidean Geometry, the sum of two vectors can be visualized as a parallelogram. The torsion tensor measures the failure of an infinitesimal parallelogram to close. That is, torsion quantifies the non-commutativity of infinitesimal displacements. Consider two curves: $x^{\mu}=x^{\mu}(\lambda)$ and $\tilde{x}^{\mu}=\tilde{x}^{\mu}(\lambda)$ with tangent vectors $u^{\mu}=\frac{dx^{\mu}}{d\lambda}$ and $\tilde{u}^{\mu}=\frac{d\tilde{x}^{\mu}}{d\lambda}$ Now, displace $u^{\alpha}$ infinitesimally along $\tilde{x}^{\mu}$: 
\begin{equation}\label{aux1}
u'^{\alpha}=u^{\alpha}+(\partial_{\mu}u^{\alpha})d\tilde{x}^{\mu}    
\end{equation}
Since $u^{\alpha}$ is parallel transported along $\tilde{x}$ it follows that
\begin{equation}
\frac{d\tilde{x}^{\mu}}{d\lambda}\partial_{\mu}u^{\alpha}+\Gamma_{\hspace{0.4em}\nu\mu}^{\alpha}\frac{d\tilde{x}^{\mu}}{d\lambda}u^{\nu}=0 \implies (\partial_{\mu}u^{\alpha})d\tilde{x}^{\mu}=-\Gamma_{\hspace{0.4em}\nu\mu}^{\alpha}u^{\nu}\tilde{u}^{\mu}d\lambda    
\end{equation}
Substituting this into equation (\ref{aux1})
\begin{equation}
    u'^{\alpha}=u^{\alpha}-\Gamma_{\hspace{0.4em}\nu\mu}^{\alpha}u^{\nu}\tilde{u}^{\mu}d\lambda
\end{equation}
Applying the same reasoning to the other vector gives
\begin{equation}
\tilde{u}'{}^{\alpha}=\tilde{u}^{\alpha}-\Gamma_{\hspace{0.4em}\nu\mu}^{\alpha}\tilde{u}^{\nu}u^{\mu}d\lambda=\tilde{u}^{\alpha}-\Gamma_{\hspace{0.4em}\mu\nu}^{\alpha}\tilde{u}^{\mu}u^{\nu}d\lambda    
\end{equation}
The failure to close the parallelogram is then quantified by
\begin{equation}
(\tilde{u}^{\alpha}+u'{}^{\alpha})-(u^{\alpha}+\tilde{u}'{}^{\alpha})=\tilde{u}^{\alpha}+u^{\alpha}-\Gamma_{\hspace{0.4em}\nu\mu}^{\alpha}u^{\nu}\tilde{u}^{\mu}d\lambda-u^{\alpha}-\tilde{u}^{\alpha}+\Gamma_{\hspace{0.4em}\mu\nu}^{\alpha}\tilde{u}^{\mu}u^{\nu}d\lambda=T_{\hspace{0.4em}\mu\nu}^{\alpha}\tilde{u}^{\mu}u^{\nu}d\lambda    
\end{equation}

The non-metricity tensor provides a measure for how the magnitude of a vector changes under parallel transport. The change in the magnitude of a vector $V^\mu$ under parallel transport along a curve $x^{\mu}=x^{\mu}(\lambda)$ is
\begin{equation}\label{change of norm}
 \frac{D}{d\lambda}(V^{\mu}V^{\nu}g_{\mu\nu})=2\frac{dx^{\alpha}}{d\lambda}(\nabla_{\alpha}V^{\mu})V_{\mu}+\frac{dx^{\alpha}}{d\lambda}(\nabla_{\alpha}g_{\mu\nu})V^{\mu}V^{\nu}   
\end{equation}
Since $V^{\mu}$ is parallel transported along the curve, $\frac{dx^{\alpha}}{d\lambda}(\nabla_{\alpha}V^{\mu})=0$ by definition (\ref{Transporte Paralelo}). Thus, we are left with
\begin{equation}\label{change of norm 2}
  \frac{D}{d\lambda}(||V^{2}||)=Q_{\alpha\mu\nu}\frac{dx^{\alpha}}{d\lambda}V^{\mu}V^{\nu}  
\end{equation}
This shows that non-metricity $Q_{\alpha\mu\nu}$ directly controls the variation of the length of a vector during parallel transport. In Riemannian geometry, where $Q_{\alpha\mu\nu}=0$, the length of a vector remains unchanged under parallel transport.

\subsection{Decomposition of the Connection}

GR adopts the Riemannian Geometry framework, where it is postulated that $T_{\hspace{0.5em}\mu\nu}^{\alpha}=0$ and $Q_{\alpha\mu\nu}=0$. 
These two geometric conditions are satisfied if and only if the connection is the Levi-Civita connection, which is uniquely determined by the metric 
\begin{equation}\label{Levi-Civita Connection}
  \Gamma_{\hspace{0.5em}\mu\nu}^{\alpha}=\frac{1}{2}g^{\alpha\lambda}(\partial_{\mu}g_{\nu\lambda}+\partial_{\nu}g_{\mu\lambda}-\partial_{\lambda}g_{\mu\nu}) 
\end{equation}
Since the Levi-Civita connection is fully determined by the metric, a spacetime in General Relativity is fully described by the pair $(\mathcal{M},g)$.

In contrast, a metric-affine geometry $(\mathcal{M},g,\Gamma)$ considers a more general connection $\Gamma$ that is not necessarily Levi-Civita and can include contributions from both torsion and non-metricity.  In this context, the connection can be decomposed into the Levi-Civita part plus extra contributions
\begin{equation}\label{decomposition of connection}
\Gamma_{\hspace{0.5em}\mu\nu}^{\alpha}=\mathring\Gamma_{\hspace{0.5em}\mu\nu}^{\alpha}+D_{\hspace{0.4em}\mu\nu}^{\alpha}
\end{equation}
where $\mathring\Gamma_{\hspace{0.5em}\mu\nu}^{\alpha}$ is the Levi-Civita connection and $D_{\hspace{0.4em}\mu\nu}^{\alpha}$ is the distortion tensor, which captures deviations from Riemannian geometry. It is composed of two parts: the contorsion tensor, which encodes torsion
\begin{equation}\label{contorsion tensor}
K_{\hspace{0.4em}\mu\nu}^{\alpha}=\frac{1}{2}g^{\alpha\lambda}(T_{\lambda\mu\nu}+T_{\mu\lambda\nu}+T_{\nu\lambda\mu}) 
\end{equation}
and the disformation tensor, which encodes non-metricity
\begin{equation}\label{disformation tensor}
L_{\hspace{0.4em}\mu\nu}^{\alpha}=\frac{1}{2}g^{\alpha\lambda}(Q_{\lambda\mu\nu}-Q_{\mu\lambda\nu}-Q_{\nu\lambda\mu})
\end{equation}
Thus, the total distortion is
\begin{equation}\label{distorsion tensor}
D_{\hspace{0.4em}\mu\nu}^{\alpha}=K_{\hspace{0.4em}\mu\nu}^{\alpha}+L_{\hspace{0.4em}\mu\nu}^{\alpha}
\end{equation}
Since the Riemann tensor depends only on the connection, this decomposition induces a corresponding decomposition of the Riemann tensor
\begin{equation}
R_{\hspace{0.5em}\beta\mu\nu}^{\alpha}(\Gamma)=\mathring{R}_{\hspace{0.5em}\beta\mu\nu}^{\alpha}+\mathring\nabla_\mu D^\alpha_{\hspace{0.5em}\nu\beta}-\mathring\nabla_\nu D^\alpha_{\hspace{0.5em}\mu\beta}+D^\sigma_{\hspace{0.5em}\nu \beta}D^\alpha_{\hspace{0.5em}\mu \sigma}-D^\sigma_{\hspace{0.5em}\mu \beta}D^\alpha_{\hspace{0.5em}\nu \sigma}       
\end{equation}
where $\mathring{R}_{\hspace{0.5em}\mu\nu\rho}^{\alpha}$ and $\mathring\nabla_{\mu}$ are the Riemann tensor and covariant derivative associated with the Levi-Civita connection.

We define the following scalar quantities:

\textit{Torsion scalar}
\begin{equation}
\mathbb{T}=-\frac{1}{4}T^{\alpha\mu\nu}T_{\alpha\mu\nu}-\frac{1}{2}T^{\alpha\mu\nu}T_{\mu\alpha\nu}+T^{\alpha}T_{\alpha}   
\end{equation}

\textit{Non-metricity scalar}
\begin{equation}
\mathbb{Q}=-\frac{1}{4}Q_{\alpha\mu\nu}Q^{\alpha\mu\nu}+\frac{1}{2}Q_{\alpha\mu\nu}Q^{\mu\alpha\nu}+\frac{1}{4}Q_{\alpha}Q^{\alpha}-\frac{1}{2}Q_{\alpha}\bar{Q}^{\alpha}
\end{equation}

Contracting the Riemann tensor decomposition leads to a decomposition of the Ricci scalar
\begin{equation}
R(\Gamma)=\mathring{R}-\mathbb{T}-\mathbb{Q}+T^{\rho\mu\nu}Q_{\mu\nu\rho}-T^{\mu}Q_{\mu}+T^{\mu}\bar{Q}_{\mu}+\mathring\nabla_{\alpha}(Q^{\alpha}-\bar{Q}^{\alpha}+2T^{\alpha})    
\end{equation}
This result connects the scalar curvature of a general connection to quantities derived from the Levi-Civita connection and the torsion and non-metricity tensors. The special cases with vanishing curvature give useful identities.

In a spacetime with zero curvature and zero non-metricity ($R(\Gamma)=0$, $Q_{\alpha\mu\nu}=0$) we obtain
\begin{equation}\label{identity in TEGR}
\mathring{R}=\mathbb{T}-2\mathring\nabla_{\alpha}T^{\alpha}   
\end{equation}

In a spacetime with zero curvature and zero torsion ($R(\Gamma)=0$, $T^\alpha_{\hspace{0.7em}\mu\nu}=0$), we obtain
\begin{equation}\label{identity in STEGR}
\mathring{R}=\mathbb{Q}-\mathring\nabla_{\alpha}(Q^{\alpha}-\bar{Q}^{\alpha})  
\end{equation}

These identities will be subsequently used to construct alternative theories of gravity.

\section{\label{sec3}Theories of Gravity}

\subsection{General Relativity}

In General Relativity, it is postulated that both torsion and non-metricity vanish,
\begin{equation}
    T_{\hspace{0.5em}\mu\nu}^{\alpha}=0 \hspace{4em} Q_{\alpha\mu\nu}=0
\end{equation} As a result, gravity manifests as the curvature of spacetime. The field equations of General Relativity relate the curvature of spacetime to the energy-momentum content. They can be derived from the Einstein–Hilbert action
 \begin{equation}
S=\frac{1}{2k}\int_{\mathcal{M}}d^{4}x\sqrt{-g}R+S_{M}  
 \end{equation}
where $S_{M}$ is the matter action and $k=8\pi G$ in natural units. Varying this action with respect to the metric yields the Einstein field equations:
\begin{equation}
G_{\mu\nu}=kT_{\mu\nu}    
\end{equation}
where the Einstein Tensor is given by
\begin{equation}
    G_{\mu\nu}=R_{\mu\nu}-\frac{1}{2}g_{\mu\nu}R\,.
\end{equation}

\subsection{Teleparalellism}

Teleparallelism was first introduced by Einstein himself, as an attempt to unify gravity with electromagnetism \cite{unzicker2005translation,bahamonde2023teleparallel}. In Teleparallel Gravity, spacetime is flat and metric-compatible and gravity is mediated by torsion. Consider a metric-affine geometry $(\mathcal{M},g,\Gamma)$ where the connection satisfies
\begin{equation}
    R_{\hspace{0.7em}\mu\nu\rho}^{\alpha}=0 \hspace{4em} Q_{\alpha\mu\nu}=0
\end{equation} 
Under these conditions, we have shown that
\begin{equation}
\mathring{R}=\mathbb{T}-2\mathring\nabla_{\alpha}T^{\alpha}  \tag{\ref{identity in TEGR}} 
\end{equation}
The dynamics is specified by the action 
\begin{equation}\label{eq t}
S_{TEGR}=\frac{1}{2k}\int_{\mathcal{M}}d^{4}x\sqrt{-g}\mathbb{T}+S_{M}
\end{equation}
From equation (\ref{identity in TEGR}), we see that this Lagrangian differs from the Einstein–Hilbert Lagrangian by a total divergence. Therefore, the metric field equations derived from (\ref{eq t}) are equivalent to those of General Relativity. This theory is known as Teleparallel Equivalent of General Relativity (TEGR).

\subsection{STG}

In 1999, Nester \& Yo \cite{nester1998symmetric} proposed an alternative theory, Symmetric Teleparallel Gravity, constructed by setting curvature and torsion to zero, with gravity mediated by non-metricity. In this case
\begin{equation}
\mathring{R}=\mathbb{Q}-\mathring\nabla_{\alpha}(Q^{\alpha}-\bar{Q}^{\alpha})\tag{\ref{identity in STEGR}}
\end{equation}
and therefore a theory described by the action
\begin{equation}\label{eq q}
S_{TEGR}=\frac{1}{2k}\int_{\mathcal{M}}d^{4}x\sqrt{-g}\mathbb{Q}+S_{M}
\end{equation}
is equivalent to GR.

\subsection{Modified Teleparallelism}

The equivalence of GR, TEGR and STEGR raises profound questions about the fundamental geometric nature of spacetime. While General Relatity describes gravity as a manifestation of spacetime curvature, TEGR and STEGR attribute gravitational effects to torsion and non-metricity, respectively. Yet, they have the same field equations and therefore it appears that these three geometric frameworks are observationally indistinguishable. This ambiguity is known as the problem of geometric underdetermination \cite{KNOX, mulder2024underdetermination,mancini2025equivalent}.

A further consequence of this degeneracy is that these equivalent theories inherit the same limitations. In particular, none of them can provide an explanation for the accelerated expansion of the Universe without resorting to a cosmological constant. If  GR requires such a constant to fit cosmological observations, then so do TEGR and STEGR.

This motivates the study of modified gravity theories using the geometric trinity of gravity as a starting point. Extending the Lagrangian to arbitrary function of the scalar $f(R),f(\mathbb{T}),f(\mathbb{Q})$ breaks this degeneracy. These modified theories could exhibit distinct phenomenologies and potentially offer alternative explanations for dark energy.  The use of modified versions of these theories, having in mind their discrimination, was explored in \cite{capozziello2024comparing} by comparing their suitability in describing inflation.

\subsubsection{$f(R)$ gravity}

A natural extension of GR, first considered by Buchdahl in 1970 \cite{Buchdahl}, is to replace the Ricci scalar $R$ in the Einstein-Hilbert action with an arbitrary function $f(R)$.  One of its main motivations is its potential to explain cosmic acceleration without invoking dark energy. Notably, Starobinsky used an $R^2$ correction to GR to model inflation \cite{Starobinsky1980}, producing what is still one of the most reliable inflation models. Additionally, $f(R)$ theories have been employed to describe the late-time acceleration of the universe \cite{CAPOZZIELLO_2002,Carroll_2004}, while \cite{Nojiri_2003} used it in an attempted unified description of early and late acceleration.

We define the theory by the action
\begin{equation}
S=\int_{\mathcal{M}}d^{4}x\frac{1}{2k}\sqrt{-g}f(R)+S_M \,, 
\end{equation}
and the $f(R)$ field equations derived in the metric formalism are:
\begin{equation}\label{metric eq f(R)}
f_RR_{\mu\nu}-\frac{1}{2}g_{\mu\nu}f(R)+(g_{\mu\nu}\square-\nabla_{\mu}\nabla_{\nu})f_R=8\pi GT_{\mu\nu}    
\end{equation}
where $\square\equiv\nabla_{\mu}\nabla^{\mu}$ and $f_R\equiv\frac{df}{dR}$. The field equations are fourth order non-linear equations for the metric, due to the second order differential operator $g_{\mu\nu}\square-\nabla_{\mu}\nabla_{\nu}$ acting on $f_R$ which already contains second order derivatives of the metric.

\subsubsection{$f(\mathbb{T})$ and $f(\mathbb{Q})$ gravity}

An analogous modification in the teleparallel framework leads to replacing the scalar $\mathbb{T}$ in TEGR with a general function $f(\mathbb{T})$ and the scalar $\mathbb{Q}$ in STEGR with a general function $f(\mathbb{Q})$. We define the torsion superpotential 
\begin{equation}\label{torsion superpotential}
S_{\alpha}^{\hspace{1em}\mu\nu}=\frac{1}{4}(T_{\hspace{0.5em}\alpha}^{\mu\hspace{0.5em}\nu}+T_{\alpha}^{\hspace{0.5em}\mu\nu}+T_{\hspace{1em}\alpha}^{\nu\mu})+\frac{1}{2}\delta_{\alpha}^{\mu}T^{\nu}-\frac{1}{2}\delta_{\alpha}^{\nu}T^{\mu}
\end{equation} 
which can be used to write the torsion scalar as
\begin{equation}
    \mathbb{T}=-S_{\alpha}^{\hspace{0.6em}\mu\nu}T_{\hspace{0.6em}\mu\nu}^{\alpha}
\end{equation}

Further defining the non-metricity superpotencial 
\begin{equation}
    P^\alpha_{\hspace{0.5em}\mu\nu}=\frac{1}{4}[-Q^\alpha_{\hspace{0.5em}\mu\nu}+Q^{\hspace{0.5em}\alpha}_{\mu\hspace{0.5em}\nu}+Q^{\hspace{0.5em}\alpha}_{\nu\hspace{0.5em}\mu}+(Q^\alpha-\bar{Q}^\alpha)g_{\mu\nu}-\frac{1}{2}(\delta^\alpha_\mu Q_\nu+\delta^\alpha_\nu Q_\mu)]
\end{equation}
so that
\begin{equation}
\mathbb{Q}=P_{\alpha\mu\nu}Q^{\alpha\mu\nu}\,,
\end{equation}
the metric field equations for these theories can be written in a form resembling GR \cite{heisenberg2024review}:
\begin{equation}\label{metric equation f(T) Einstein form}
  f_\mathbb{T}G_{\mu\nu}-\frac{1}{2}g_{\mu\nu}(f-f_\mathbb{T}\mathbb{T})+2f_{\mathbb{T}\mathbb{T}}(\mathbb{T})S_{(\mu\nu)}^{\hspace{1.7em}\alpha}\partial_{\alpha}\mathbb{T}=kT_{\mu\nu}  
\end{equation}
\begin{equation}\label{metric equation f(Q) Einstein form}
f_\mathbb{Q}G_{\mu\nu}-\frac{1}{2}(f-f_\mathbb{Q}\mathbb{Q})+2f_{\mathbb{Q}\mathbb{Q}}P_{\hspace{0.5em}\mu\nu}^{\alpha}\partial_{\alpha}\mathbb{Q}=kT_{\mu\nu}
\end{equation}
From this form, it becomes clear that the cases $f(\mathbb{T})=\mathbb{T}$ and $f(\mathbb{Q})=\mathbb{Q}$ are equivalent to General Relativity.
The connection field equations are
\begin{subequations}
\label{connection equations}
\begin{align}
\label{connection equation f(T)}
(\nabla_{\mu}+T_{\mu})(f_\mathbb{T}S_{[\alpha\hspace{0.5em}\beta]}^{\hspace{0.5em}\mu})&=0\\
\nabla_{\mu}\nabla_{\nu}(\sqrt{-g}f_\mathbb{Q}P_{\hspace{1em}\alpha}^{\mu\nu})&=0
\end{align}
\end{subequations}

\subsubsection{Boundary Term Extensions}

The teleparallel theories from the modified trinity can be extended to include $f(\mathbb{T})$, $f(\mathbb{Q})$ and $f(R)$ gravity as special cases. This is achieved by promoting the boundary terms in the equivalence relations (\ref{identity in TEGR}) and (\ref{identity in STEGR}) to independent variables within the action.

The equivalence between GR and TEGR was a result of the identity
\begin{equation}
\mathring{R}=\mathbb{T}-2\mathring\nabla_{\alpha}T^{\alpha}  \tag{\ref{identity in TEGR}}\,.
\end{equation}
Defining the boundary term $B\equiv2\mathring\nabla_{\alpha}T^{\alpha}$, we consider the action
\begin{equation}
S=\frac{1}{2k}\int_{\mathcal{M}}d^{4}x\sqrt{-g}f(\mathbb{T},B)+S_{M}    
\end{equation}
The corresponding metric field equations, derived in \cite{capozziello2025extended}, are
\begin{equation}\label{equacao f(T,B)}
G_{\mu\nu}f_{\mathbb{T}}-\frac{1}{2}g_{\mu\nu}(f-f_{\mathbb{T}}\mathbb{T}-f_{B}B)+2\partial_{\alpha}(f_{\mathbb{T}}+f_{B})S_{(\mu\nu)}^{\hspace{1em}\alpha}+\nabla_{\mu}\nabla_{\nu}f_{B}-g_{\mu\nu}\square f_{B}=kT_{\mu\nu} 
\end{equation}
If $f(\mathbb{T},B)=f(\mathbb{T})$ we get back the field equation for $f(\mathbb{T})$ gravity (\ref{metric equation f(T) Einstein form}). If $f(\mathbb{T},B)=f(\mathbb{T}-B)=f(R)$, which implies the relation $f_\mathbb{T}=f_R=-f_B$, we get back the field equation for $f(R)$ gravity (\ref{metric eq f(R)}). 

The connection field equation is
\begin{equation}\label{equacao conexao f(T,C)}
(\nabla_{\alpha}+T_{\alpha})[S_{[\mu\hspace{0.5em}\nu]}^{\hspace{0.7em}\alpha}\sqrt{-g}(f_\mathbb{T}+f_{B})]=0    
\end{equation}
which reduces to the equation (\ref{connection equation f(T)}) in the case $f(\mathbb{T},B)=f(\mathbb{T})$ and becomes an identity in the case $f(\mathbb{T},B)=f(\mathbb{T}-B)=f(R)$, consistent with the absence of a dynamical connection in metric $f(R)$.

Analogously for STG, we use the identity
\begin{equation}
\mathring{R}=\mathbb{Q}-\mathring\nabla_{\alpha}(Q^{\alpha}-\bar{Q}^{\alpha})  
\tag{\ref{identity in STEGR}}
\end{equation}
to define a boundary term variable $C\equiv\mathring\nabla_{\alpha}(Q^{\alpha}-\bar{Q}^{\alpha})$ and construct the action
\begin{equation}
  S=\frac{1}{2k}\int_{\mathcal{M}}d^{4}x\sqrt{-g}f(\mathbb{Q},C)+S_{M}  
\end{equation}
The metric field equations, derived in \cite{capozziello2023role}, are
\begin{equation}\label{equacao f(Q,C)}
G_{\mu\nu}f_{\mathbb{Q}}-\frac{1}{2}g_{\mu\nu}(f-f_{\mathbb{Q}}\mathbb{Q}-f_{C}C)+2\partial_{\alpha}(f_{\mathbb{Q}}+f_{C})P_{\hspace{0.5em}\mu\nu}^{\alpha}+\nabla_{\mu}\nabla_{\nu}f_{C}-g_{\mu\nu}\square f_{C}=kT_{\mu\nu} 
\end{equation}
Again if $f(\mathbb{Q},B)=f(\mathbb{Q})$ we get back the field equation for $f(\mathbb{Q})$ gravity (\ref{metric equation f(Q) Einstein form}) and if $f(\mathbb{Q},C)=f(\mathbb{Q}-C)=f(R)$ we get back the field equation for $f(R)$ gravity (\ref{metric eq f(R)}).
The connection field equation is
\begin{equation}\label{equacao conexao f(Q,C)}
    \nabla_{\mu}\nabla_{\nu}[P_{\hspace{1em}\alpha}^{\mu\nu}\sqrt{-g}(f_{\mathbb{Q}}+f_{C})]=0
\end{equation}
which reduces to connection field equation of $f(\mathbb{Q})$ in the appropriate limit.

\subsection{Connection in Teleparallelism}\label{connection in tele}

Contrary to General Relativity where the geometric postulates fix the connection to be the Levi connection, which is entirely determined by the metric, in these new theories the geometric postulates do not completely fix the connection. In both TG and STG, the curvature is zero. We impose the teleparallel constraint by making the Riemann tensor vanish
\begin{equation}
R_{\hspace{0.4em}\mu\nu\rho}^{\alpha}=2\partial_{[\nu}\Gamma_{\hspace{0.4em}\rho]\mu}^{\alpha}+2\Gamma_{\hspace{0.4em}[\nu|\lambda}^{\alpha}\Gamma_{\hspace{0.4em}\rho]\mu}^{\lambda}=0
\end{equation}
which has the solution
\begin{equation}\label{flat connection}
\Gamma_{\mu\nu}^{\alpha}=(\Lambda^{-1})_{\lambda}^{\alpha}\partial_{\mu}\Lambda_{\nu}^{\lambda}   
\end{equation}
where $\Lambda_{\nu}^{\mu}$ is a matrix belonging to $GL(4,\mathbb{R})$. Any connection of this form is flat.

Imposing metric compatibility to this connection leads to
\begin{equation}
g^{\lambda(\mu}\partial_{\alpha}\Lambda_{\rho}^{\nu)}(\Lambda^{-1})_{\lambda}^{\rho}=\frac{1}{2}\partial_{\alpha}g^{\mu\nu}     
\end{equation}
This relation can be used to eliminate the metric in terms of the matrix $\Lambda$. Imposing instead that the flat connection (\ref{flat connection}) is symmetric, we get:
\begin{equation}
\partial_{[\mu}\Lambda_{\nu]}^{\alpha}=0\implies\Gamma_{\mu\nu}^{\alpha}=\frac{\partial x^{\alpha}}{\partial\xi^{\lambda}}\partial_{\mu}\partial_{\nu}\xi^{\lambda}    
\end{equation}
For $\xi^{\mu}=M_{\nu}^{\mu}x^{\nu}+\xi_{c}^{\mu}$ where $M$ is a constant matrix, the connection vanishes. This choice corresponds to the coincident gauge. There always exists a coordinate transformation that brings the connection into this trivial form. This can be used to construct a formulation of general relativity based on non-metricity, but without the affine structure, called Coincident General Relativity \cite{heisenberg2024review,jimenez2018coincident}.

A similar gauge can be constructed for TEGR using a different formalism. Instead of working with the metric and connection, one uses the tetrad and the spin connection. In a flat and metric-compatible spacetime, there always exists a local Lorentz transformation that makes the spin connection vanish. This is known as the Weitzenböck gauge. While it might seem that such gauges simplify calculations, this is not generally the case. The coordinate system in which the connection vanishes may not coincide with the one best adapted to the symmetries of the spacetime. Although the coordinate transformation to the coincident gauge eliminates the affine connection, the same information is transferred into the components of the metric, which becomes more complicated as a result. In other words, a diffeomorphism that trivializes the connection will typically make the metric more complex. All the information that was encoded in the symmetry-reduced connection is now moved into the metric. Hence, nothing is truly gained by using the coincident gauge.

For this reason, we will not use the coincident gauge or the Weitzenböck gauge in what follows and won't develop the tetrad formalism for TG. Instead, we adopt a metric-affine approach, where the metric and the connection are treated as independent. The connection is then constrained by the geometric postulates of the theory, and both the metric and the connection are required to be compatible with the symmetries of the spacetime under study.

\section{\label{sec4}Cosmology}

In general relativity, the cosmological principle, which posits that the Universe is homogeneous and isotropic on large scales, leads to the Friedmann-Lemaitre-Robertson-Walker (FLRW) metric
\begin{equation}
    ds^{2}=-dt^2+a^{2}(t)\left(\frac{dr^2}{1-\kappa r^2}+r^{2}d\Omega^{2}\right)
\end{equation}
where $\kappa$ denotes the spatial curvature The matter content of the Universe is modeled as a perfect fluid with energy-momentum tensor
\begin{equation}
    T_{\mu}^{\nu}=diag(-\rho,p,p,p)
\end{equation}
where $\rho$ is the energy density and $p$ the pressure of the fluid.

Substituting the FLRW metric and energy-momentum tensor into the Einstein field equations with a cosmological constant, we obtain the Friedmann equations
\begin{subequations}
    \begin{align}
        3H^{2}+\frac{3\kappa}{a^2}&=k\rho+\Lambda \label{Friedmann GR}\\
        2\dot{H}+3H^2+\frac{\kappa}{a^2}&=-kp+\Lambda
    \end{align}
\end{subequations}
In $\Lambda$CDM, the recent universe is spatially flat and has significant amounts of dust-like matter and dark energy. The first Friedmann equation for such composition reads
\begin{equation}\label{LDCM equation}
    3H^2=k\rho+\Lambda
\end{equation}
where the matter density evolves as $\rho=\frac{\rho_0}{a^3}$, following from the conservation equation for pressureless matter. In what follows, we aim to reproduce this solution within the alternative geometric frameworks of TG and STG. Our goal is to recover (\ref{LDCM equation}) but without introducing a cosmological constant explicitly in the action.

\subsection{Symmetry Reduction in Teleparallelism\label{symmetry cosmology}\label{reduction41}}

In contrast to GR or metric $f(R)$, where the geometric postulates of vanishing torsion and metric compatibility completely determine the affine connection in terms of the metric, TG and STG do not fully constrain the connection. While it is always possible to choose coordinates in which the connection vanishes identically—namely, the Weitzenböck gauge in TG and the coincident gauge in STG—these coordinate systems may not be naturally adapted to the symmetries of the spacetime under consideration.

For instance, in cosmological settings, the standard FLRW metric expressed in spherical coordinates cannot be consistently combined with the coincident gauge. Implementing the coincident gauge for that metric requires transforming the FLRW metric to Cartesian coordinates. Furthermore, adopting such gauges can obscure alternative connections that are also compatible with the geometric postulates. For these reasons, there has been a recent shift toward abandoning the coincident gauge in favor of a fully covariant formulation of $f(\mathbb{T})$ \cite{krssak2016covariant} and $f(\mathbb{Q})$ gravity \cite{zhao2022covariant}, where the metric and connection are treated as a priori independent variables. In this approach, the connection is subsequently constrained by the geometric postulates and by the imposition of spacetime symmetries.

In a cosmological context, this amounts to requiring that both the metric and the connection be invariant under the symmetries of spatial homogeneity and isotropy. Mathematically, this symmetry requirement is enforced by demanding that the Lie derivatives of the metric and the connection with respect to the Killing vectors generating rotations and translations vanish
\begin{subequations}\label{Lie Derivative}
    \begin{align}
    \mathcal{L}_{\xi}g^{\mu\nu}&=0\\
    \mathcal{L}_{\xi}\Gamma_{\hspace{0.7em}\mu\nu}^{\alpha}&=0
    \end{align}
\end{subequations}
The imposition of these symmetry constraints on the metric yields the FLRW metric
\begin{equation}
    ds^{2}=g_{tt}dt^2+a^{2}(t)\left(\frac{dr^{2}}{1-\kappa r^{2}}+r^{2}d\Omega^{2}\right)
\end{equation}
where $\kappa$ denotes the spatial curvature. The component $g_{tt}$ can always be set to $g_{tt}=-1$ by an appropriate coordinate transformation. For the remainder of this work, we will adopt this convention and fix $g_{tt}=-1$. Furthermore, since observational data strongly supports a spatially flat universe, we will assume the spatial curvature $\kappa=0$ from this point onward. {We are left with a single free function $a(t)$ to be determined by the field equations.}

{For the affine connection, the symmetry constraints (\ref{Lie Derivative}) reduce its 64 components to five arbitrary functions of time. To further decrease the number of degrees of freedom, it is then necessary to solve the equations arising from the geometric postulates of each theory.} For both $f(\mathbb{T})$ and $f(\mathbb{Q})$ gravity, we impose vanishing curvature
\begin{equation}
    R_{\hspace{0.4em}\mu\nu\rho}^{\alpha}=0
\end{equation}
In addition, we impose vanishing torsion $T_{\hspace{0.7em}\mu\nu}^{\alpha}=0$ for $f(\mathbb{Q})$ and for $f(\mathbb{T})$ vanishing non-metricity $Q_{\alpha\mu\nu}=0$. Since the torsion tensor is antisymmetric in its lower indices,the condition $T_{\hspace{0.7em}\mu\nu}^{\alpha}=0$  imposes 24 independent constraints, reducing the number of independent connection components from 64 to 40. In contrast, the non-metricity tensor is symmetric in its last two indices, so $Q_{\alpha\mu\nu}=0$  imposes 40 independent constraints. Thus, we expect that the STG connection is more strongly constrained than in the torsional case.
These equations are numerous and a systematic derivation has been carried out in \cite{Heisenberg_cosmology}. After imposing both the geometric and symmetry conditions, the following results emerge: In $f(\mathbb{T})$ gravity, there exists a unique spatially flat connection compatible with the FLRW metric, denoted $\Gamma^{0}$, which is presented in equation (\ref{Connection f(T) k=0}). If spatial curvature is permitted, four distinct classes of connections arise, all of which converge to $\Gamma^{0}$ in limit $\kappa\rightarrow0$. In all cases, the connection is completely determined by the metric, indicating that no additional degrees of freedom are introduced by the connection in this theory.

Connection $\Gamma^0$

\begin{equation} \label{Connection f(T) k=0}
\begin{aligned}
&\Gamma^t_{\hspace{0.5em}\mu\nu}=\bold{0}^{(4\times4)}\hspace{7em}
&\Gamma^r_{\hspace{0.5em}\mu\nu}=\begin{pmatrix} 
0 & H & 0 &0\\
0 & 0 & 0 &0\\
0 & 0 & \frac{1}{r} &0\\
0 & 0 & 0 &-r  \sin^2\theta
\end{pmatrix}\\
&\Gamma^\theta_{\hspace{0.5em}\mu\nu}=\begin{pmatrix}
0 & 0 & H &0\\
0 & 0 & \frac{1}{r} &0\\
0 &  \frac{1}{r}&0 &0\\
0 & 0 & 0 &-\cos\theta\sin\theta
\end{pmatrix}
&\Gamma^\phi_{\hspace{0.5em}\mu\nu}=\begin{pmatrix}
0 & 0 & 0 &H\\
0 & 0 & 0 &\frac{1}{r}\\
0 & 0 & 0 &\cot\theta\\
0 & \frac{1}{r} & \cot\theta &0
\end{pmatrix}\\
\end{aligned}
\end{equation}

In $f(\mathbb{Q})$ gravity, three spatially flat connections are compatible with the symmetry and geometric constraints. {All three contain an unspecified function $\gamma(t)$, analogous to $a(t)$ in the metric.} Whether this function represents a propagating degree of freedom depends on the behavior of the connection field equations.
It turn out that for connection $\Gamma^1$, the connection field equation becomes an identity, and $\gamma(t)$ is left undetermined and can be set to zero. 
However, for the connections $\Gamma^2$ and $\Gamma^3$, $\gamma(t)$ remains dynamical and may correspond to an additional propagating degree of freedom \cite{Heisenberg_cosmology}.


Connection $\Gamma^1$

\begin{equation}
\label{Connection f(Q) k=0 I}
\begin{aligned}
&\Gamma^t_{\hspace{0.5em}\mu\nu}=\begin{pmatrix}
\gamma & 0 & 0 &0\\
0 & 0 & 0 &0\\
0 & 0 & 0 &0\\
0 & 0 & 0 &0
\end{pmatrix}
&\Gamma^r_{\hspace{0.5em}\mu\nu}=\begin{pmatrix}
0 & 0 & 0 &0\\
0 & 0 & 0 &0\\
0 & 0 & -r &0\\
0 & 0 & 0 &-r  \sin^2\theta
\end{pmatrix}\\
&\Gamma^\theta_{\hspace{0.5em}\mu\nu}=\begin{pmatrix}
0 & 0 & 0 &0\\
0 & 0 & \frac{1}{r} &0\\
0 &  \frac{1}{r}&0 &0\\
0 & 0 & 0 &-\cos\theta\sin\theta
\end{pmatrix}
&\Gamma^\phi_{\hspace{0.5em}\mu\nu}=\begin{pmatrix}
0 & 0 & 0 &0\\
0 & 0 & 0 &\frac{1}{r}\\
0 & 0 & 0 &\cot\theta\\
0 & \frac{1}{r} & \cot\theta &0
\end{pmatrix}\\
\end{aligned}
\end{equation}

Connection $\Gamma^2$

\begin{equation}
\label{Connection f(Q) k=0 II}
\begin{aligned}
&\Gamma^t_{\hspace{0.5em}\mu\nu}=\begin{pmatrix}
\gamma+\frac{\dot{\gamma}}{\gamma} & 0 & 0 &0\\
0 & 0 & 0 &0\\
0 & 0 & 0 &0\\
0 & 0 & 0 &0
\end{pmatrix}
&\Gamma^r_{\hspace{0.5em}\mu\nu}=\begin{pmatrix}
0 & \gamma & 0 &0\\
\gamma & 0 & 0 &0\\
0 & 0 & -r &0\\
0 & 0 & 0 &-r  \sin^2\theta
\end{pmatrix}\\
&\Gamma^\theta_{\hspace{0.5em}\mu\nu}=\begin{pmatrix}
0 & 0 & \gamma &0\\
0 & 0 & \frac{1}{r} &0\\
\gamma &  \frac{1}{r}&0 &0\\
0 & 0 & 0 &-\cos\theta\sin\theta
\end{pmatrix}
&\Gamma^\phi_{\hspace{0.5em}\mu\nu}=\begin{pmatrix}
0 & 0 & 0 &\gamma\\
0 & 0 & 0 &\frac{1}{r}\\
0 & 0 & 0 &\cot\theta\\
\gamma & \frac{1}{r} & \cot\theta &0
\end{pmatrix}
\end{aligned} 
\end{equation}

\vspace{2cm}
Connection $\Gamma^3$

\begin{equation}
\begin{aligned}
\label{Connection f(Q) k=0 III}
&\Gamma^t_{\hspace{0.5em}\mu\nu}=\begin{pmatrix}
-\frac{\dot{\gamma}}{\gamma} & 0 & 0 &0\\
0 & \gamma & 0 &0\\
0 & 0 & \gamma r^2 &0\\
0 & 0 & 0 &\gamma r^2 \sin^2\theta
\end{pmatrix}
&\Gamma^r_{\hspace{0.5em}\mu\nu}=\begin{pmatrix}
0 & 0 & 0 &0\\
0 & 0 & 0 &0\\
0 & 0 & -r &0\\
0 & 0 & 0 &-r  \sin^2\theta
\end{pmatrix}\\
&\Gamma^\theta_{\hspace{0.5em}\mu\nu}=\begin{pmatrix}
0 & 0 & 0 &0\\
0 & 0 & \frac{1}{r} &0\\
0 &  \frac{1}{r}&0 &0\\
0 & 0 & 0 &-\cos\theta\sin\theta
\end{pmatrix}
&\Gamma^\phi_{\hspace{0.5em}\mu\nu}=\begin{pmatrix}
0 & 0 & 0 &0\\
0 & 0 & 0 &\frac{1}{r}\\
0 & 0 & 0 &\cot\theta\\
0 & \frac{1}{r} & \cot\theta &0
\end{pmatrix}
\end{aligned}
\end{equation}

\subsection{Friedmann Equations for Teleparallelism}

In $f(\mathbb{T})$ gravity, the imposition of homogeneity and isotropy for a spatially flat universe uniquely determines the affine connection to be that given in (\ref{Connection f(T) k=0}). For this connection, the torsion scalar is found to be
\begin{equation}
    \mathbb{T}=-6H^2
\end{equation}
and the boundary term is given by
\begin{equation}
    B=-6(\dot{H}+3H^{2})\,.
\end{equation}
The components $tt$ and $rr$ of the $f(\mathbb{T})$ metric field equation (\ref{metric equation f(Q) Einstein form}) give the modified Friedmann equations
\begin{subequations}
    \begin{align}
        \frac{f}{2}-\mathbb{T}f_\mathbb{T}=k\rho\\
        -\frac{f}{2}+f_\mathbb{T}(\mathbb{T}-2\dot{H})-2H\partial_{t}\mathbb{T}f_{\mathbb{T}\mathbb{T}}=kp
    \end{align}
\end{subequations}
{It is often convenient to rewrite the action as $f(\mathbb{T})=\mathbb{T}+F(\mathbb{T})$ where $F$ encodes deviations from GR. This formulation allows the Friedmann equations to be expressed in a form closer to that of GR, thereby facilitating comparison between the theories.}
The first Friedmann equation then becomes
\begin{equation}
    3H^{2}=k\rho+F_\mathbb{T}\mathbb{T}-\frac{F}{2}
\end{equation}

In comparison to the standard Friedmann equation in GR (\ref{Friedmann GR}), the function $F(\mathbb{T})$ acts as a source to the matter sector. This motivates the definition of an effective total energy density: \begin{equation}
    k\rho_{eff}\equiv k\rho+F_\mathbb{T}\mathbb{T}-\frac{F}{2}
\end{equation} 
so that the Friedmann equation assumes its standard form
\begin{equation}
    3H^2=k\rho_{eff}
\end{equation}

{This effective fluid is, of course, not a real material field; rather, it indicates that the geometry behaves as if additional matter sources were present.}

We may also define a dark energy density contribution as 
\begin{equation}\label{dark energy f(T)}
    \rho_{DE}=F_\mathbb{T}\mathbb{T}-\frac{F}{2}
\end{equation} 
such that $\rho_{eff}=\rho+\rho_{DE}$. {This term will be the source of an effective cosmological constant. Consequently, expressing the Friedmann equation in the effective fluid formalism facilitates both the reconstruction procedure and the identification of the terms responsible for dark energy.}

In the case of $f(\mathbb{T},B)$ the modified Friedmann equations become
\begin{subequations}
    \begin{align}
        \frac{f}{2}-f_{\mathbb{T}}\mathbb{T}-\frac{1}{2}f_{B}B-3H\dot{f_{B}}&=k\rho \\
        -\frac{f}{2}+f_{\mathbb{T}}\mathbb{T}+\frac{1}{2}f_{B}B-2\dot{H}f_{\mathbb{T}}-2H\dot{f_{\mathbb{T}}}+\ddot{f_{B}}&=kp
    \end{align}
\end{subequations}
Assuming again the decomposition $f(\mathbb{T},B)=\mathbb{T}+F(\mathbb{T},B)$, we obtain the effective energy density
\begin{equation}
    3H^{2}=k\rho_{eff}\equiv k\rho+TF_{\mathbb{T}}-\frac{F}{2}+\frac{1}{2}F_{B}B+3H\dot{F_{B}}
\end{equation}
In this framework, the dark energy contribution reads:
\begin{equation}
\rho_{DE}\equiv TF_{\mathbb{T}}-\frac{F}{2}+\frac{1}{2}F_{B}B+3H\dot{F_{B}}    
\end{equation}

Similarly, for $f(\mathbb{Q},C)$ gravity, with the choice of connection $\Gamma^1$, one finds
\begin{subequations}
    \begin{align}
    Q&=-6H^2\\
    C&=-6(3H^2+\dot H)
    \end{align}
\end{subequations}
and the Friedmann equations become
\begin{subequations}
    \begin{align}
        \frac{f}{2}-f_{\mathbb{Q}}\mathbb{Q}-\frac{1}{2}f_{C}C-3H\dot{f_{C}}&=k\rho \\
        -\frac{f}{2}+f_{\mathbb{Q}}\mathbb{Q}+\frac{1}{2}f_{C}C-2\dot{H}f_{\mathbb{Q}}-2H\dot{f_{\mathbb{Q}}}+\ddot{f_{C}}&=kp
    \end{align}
\end{subequations}
We see that for connection $\Gamma^0$ in TG and $\Gamma^1$ in STG, the structure of the Friedmann equations in  $f(\mathbb{Q},C)$ coincides with $f(\mathbb{T},B)$. This a case of correspondence between connection $\Gamma^0$ and $\Gamma^1$. Two affine connection are said to be correspondent if the non-metricity scalar in $f(\mathbb{Q})$ has the same value as the torsion scalar for $f(\mathbb{T})$ and the equations of motion take the same form regardless of $f$ \cite{wu2024correspondence}. 

However, for the alternative connections $\Gamma^2$ and $\Gamma^3$ additional dynamical degrees of freedom provide a richer cosmology to STG. For the spatially flat FLRW metric, the $f(\mathbb{T})$ solutions are just a subset of $f(\mathbb{Q})$. For the connection $\Gamma^2$, the non-metricity scalar is
\begin{equation}\label{scalar connection 2}
    \mathbb{Q}=\frac{3}{2}(6h\gamma+2\dot{\gamma}-4H^2)
\end{equation}
and the Friedmann Equations are
\begin{subequations}
    \begin{align}\label{Friedmann connection 2}
        \frac{3}{4}f_{\mathbb{Q}}(-6H\gamma-2\dot{\gamma}+8H^2)+\frac{1}{2}(3\gamma f_{\mathbb{Q}\mathbb{Q}}\dot{\mathbb{Q}}+f)&=k\rho\\
        \frac{f_\mathbb{Q}}{4}(18H\gamma+6\dot{\gamma}-8\dot{H}-24H^2)-\frac{1}{2}(f_{\mathbb{Q}\mathbb{Q}}\dot{Q}(4H-3\gamma)+f)&=kp
    \end{align}
\end{subequations}
The only non-identical connection equation is
\begin{equation}\label{connection 2 equation}
    -\frac{3\gamma}{4}(f_{\mathbb{Q}\mathbb{Q}}\dot{Q}6H+2f_{\mathbb{Q}\mathbb{Q}\mathbb{Q}}\dot{\mathbb{Q}}^{2}+2f_{\mathbb{Q}\mathbb{Q}}\ddot{\mathbb{Q}})=0
\end{equation}

For connection $\Gamma^3$, the non-metricity scalar and the metric and connection field equations are
\begin{subequations}
    \begin{align}
    &\mathbb{Q}=\frac{3}{2}(\gamma6H+2\dot{\gamma}-4H^2)\\
        &\frac{3}{4}f_\mathbb{Q}(6H\gamma+2\dot\gamma-8H^2)+\frac{1}{2}(3\gamma f_{\mathbb{Q}\mathbb{Q}}\dot{Q}-f)=-k\rho\\
        &\frac{f_\mathbb{Q}}{4}(18H\gamma+6\gamma-8\dot{H}-24H^2)-\frac{1}{2}(f_{\mathbb{Q}\mathbb{Q}}\dot{Q}(4H-\gamma)+f)=kp\\
        &\frac{3\gamma}{4}(f_{\mathbb{Q}\mathbb{Q}}\dot{Q}10H)+2f_{\mathbb{Q}\mathbb{Q}\mathbb{Q}}\dot{Q}^2+2f_{\mathbb{Q}\mathbb{Q}}\ddot{Q})+3\dot{\gamma}f_{\mathbb{Q}\mathbb{Q}}\dot{Q}=0
    \end{align}
\end{subequations}

These equations were first derived in \cite{Black_Hole}. Since then, analytical and numerical solutions have been explored in \cite{dimakis_degrees_of_freedom, ayuso_degrees_of_freedom}, the stability of the solutions was investigated in \cite{guzman2024exploring} and observation constraints studied in \cite{yang_degress_of_freedom,shi2023degrees}.

\subsection{Reconstruction}

To study the background $\Lambda$CDM cosmology, we apply the reconstruction method, which constructs the gravitational Lagrangian based on a given cosmological behavior stemming from a specific choice of the scale factor and $H(t)$. This method was pioneered in \cite{Nojiri_Odintsov2009}, which showed that for any given scale factor  $a(t)$ and  matter composition of the universe there exists a corresponding $f(R)$ theory that admits $a(t)$ as a solution. Reconstruction of $\Lambda$CDM in the context of $f(R)$ was investigated in \cite{LCDM_Reconstruction}. The reconstruction method has since been extended to a variety of modified theories, including scalar-tensor theories \cite{Scalar_Field}, Gauss-Bonnet gravity \cite{Gauss_Bonnet} and more recently, to TG and STG
\cite{gadbail2025_thesis,Gadbail_2023_Q_T,Gadbail_2023_Q_C,Caruana_2020,nojiri2024,Gadbail_2022_Q,esposito2022}.

The general reconstruction procedure is as follows: given a desired $a(t)$, one computes the relevant scalar quantity ($R$, $\mathbb{T}$ or $\mathbb{Q}$), as a function of $a(t)$, and inverts this relation to express the scale factor as a function of the scalar.
Substituting this into the modified Friedmann equation transforms it into a differential equation for the function $f$. Solving this equation yields the class of gravitational Lagrangians that reproduce the chosen cosmological behavior. In many cases, the reconstructed equations are not analytically solvable or result in very complex Lagrangians, restricting their usefulness beyond the background cosmological scale. The reconstruction of $\Lambda$CDM in the context of $f(R)$ was investigated in \cite{LCDM_Reconstruction}, where it was proved that no non-trivial $f(R)$ reproduces the exact $\Lambda$CDM background for a universe filled with dust-like matter.

\subsubsection{Reconstruction in $f(\mathbb{T})$ gravity}

We want to reproduce the effect of the cosmological constant in $f(\mathbb{T})$ gravity. Using the $\Lambda$CDM differential equation (\ref{LDCM equation}) and the value of the scale factor $\mathbb{T}=-6H^2$, we can write
\begin{equation}
    k\rho=3H^2-\Lambda=-\frac{\mathbb{T}}{2}-\Lambda
\end{equation}
and therefore, the first Friedmann equation becomes
\begin{equation}
\frac{f}{2}-\mathbb{T}f_\mathbb{T}=-\frac{\mathbb{T}}{2}-\Lambda
\end{equation} Introducing the function $F$, previously defined as $f(\mathbb{T})=\mathbb{T}+F(\mathbb{T})$, the Friedmann equation can be written as
\begin{equation}
    \mathbb{T}{F_\mathbb{T}}-\frac{F}{2}=\Lambda
\end{equation}
This equation could have been immediately obtained from $\rho_{DE}=\Lambda$, using the definition of  $\rho_{DE}$ in (\ref{dark energy f(T)}). The general solution is
\begin{equation}
    F(\mathbb{T})=-2\Lambda+\alpha\sqrt{-\mathbb{T}}
\end{equation}
We see that the cosmological constant is still necessary and the additional term $\alpha\sqrt{-\mathbb{T}}$ corresponds to the solution of the homogeneous part of the Friedmann equation and therefore does not influence the background.

\subsubsection{Reconstruction in $f(\mathbb{T},B)$ gravity}

Applying the same strategy of equating $\rho_{DE}=\Lambda$,
the reconstruction equation for $f(\mathbb{T},B)$ is
\begin{equation}\label{aux2}
    TF_{\mathbb{T}}-\frac{F}{2}+\frac{1}{2}F_{B}B+3H\dot{F}_B=\Lambda  
\end{equation}
Reconstruction cannot be straightforwardly applied to a general $F(\mathbb{T},B)$. The reason is that the method, as introduced earlier, relies on expressing all relevant quantities as functions of a single scalar argument of the Lagrangian. This is feasible when there is only one scalar. However, in $f(\mathbb{T},B)$, the torsion scalar and the boundary term are, in principle, independent variables. 
To proceed, one must restrict the functional form of $f(\mathbb{T},B)$ so that the reconstruction equation can be treated consistently without compromising the independence of $\mathbb{T}$ and $B$.
We will use as an ansatz 
\begin{equation}
F(\mathbb{T},B)=g(\mathbb{T})+Bh(\mathbb{T})\,,    
\end{equation}
from which we obtain, when substituting in the reconstruction equation (\ref{aux2})
\begin{equation}\label{reconstruction equation reduced}
    \mathbb{T}g_{\mathbb{T}}-\frac{g}{2}+3\mathbb{T}^{2}h_{\mathbb{T}}=\Lambda
\end{equation}
For any function $g$ there is always a function $h$ solving (\ref{reconstruction equation reduced}) that reconstructs $\Lambda$CDM.
For instance, choosing the simplest case $g(\mathbb{T})=0$, the solution is:
\begin{equation}
    h(\mathbb{T})=-\frac{\Lambda}{3\mathbb{T}}+c_1
\end{equation}
The arbitrary constant $c_1$ is irrelevant as it will corresponds to a term $c_1B$ in the Lagrangian, which is just a boundary term. Without it the resulting Lagrangian reads
\begin{equation}
        f(\mathbb{T},B)=\mathbb{T}-\frac{\Lambda B}{3\mathbb{T}}
\end{equation}
This model reproduces the effect of a cosmological constant through the coupling between B and $\mathbb{T}$ in the term $\frac{\Lambda B}{3\mathbb{T}}$.

The result generalizes for $g(\mathbb{T})=\alpha\mathbb{T}$, yielding
\begin{equation}\label{generalization}
    f(\mathbb{T})=\alpha\mathbb{T}-\frac{B}{3}\left(\frac{\Lambda}{\mathbb{T}}+\frac{\alpha-1}{2}\log\left(-\frac{\mathbb{T}}{\Lambda}\right)\right)
\end{equation}
The chosen ansatz is seen to be highly flexible, allowing for an infinite number of solutions while maintaining their relative simplicity compared to the more complex functions (such as hypergeometric functions) that frequently emerge from reconstruction methods. Moreover, if the function $g$ is free of integration constants, the resulting Lagrangian will never contain a cosmological constant term.

\subsubsection{Reconstruction in $f(\mathbb{Q})$ and $f(\mathbb{Q},C)$ gravity}

The above conclusions apply \textit{mutatis mutandis} to STG, when using the connection $\Gamma^1$: in $f(\mathbb{Q)}$ gravity, one again requires an explicit cosmological constant with the correction $\sqrt{-\mathbb{Q}}$ having no effect on the background; in $f(\mathbb{Q},C)$ the role of the cosmological constant can be mimicked by a term $\frac{C}{\mathbb{Q}}$. From this perspective, TG and STG yield the same cosmological behavior. Consequently, the geometric identity of spacetime (torsion vs. non-metricity) remains observationally indistinguishable. This degeneracy arises due to the specific connections chosen.
However, STG allows for a broader class of connections, including those that introduce additional degrees of freedom. We will use these extra degrees of freedom to reconstruct $\Lambda$CDM in $f(\mathbb{Q)}$.

For that, we start by proving the following result.
\begin{proposition}
\label{theorem constant scalar}
    Given an arbitrary function $f(\mathbb{Q})$, if $\mathbb{Q}$ is constant, the metric field equations of $f(\mathbb{Q})$ gravity reduce to the field equations of General Relativity.
\end{proposition}
\begin{proof}
    Let us denote the constant value of the non-metricity scalar by $\mathbb{Q}_0$. 
    Then, the field equation (\ref{metric equation f(Q) Einstein form}) simplifies to
\begin{equation}
  G_{\mu\nu}-\frac{1}{2}\frac{f(\mathbb{Q}_0)-f'(\mathbb{Q}_0)\mathbb{Q}_0}{f'(\mathbb{Q}_0)}g_{\mu\nu}=k\frac{T_{\mu\nu}}{f'(\mathbb{Q}_0)}  
\end{equation}
This corresponds to the Einstein Field Equations with  rescaled energy-momentum tensor and effective cosmological constant given by:
\begin{subequations}
\begin{align}
\tilde{T}^{\mu\nu}&=\frac{T_{\mu\nu}}{f'(\mathbb{Q}_0)}\\
\Lambda_{eff}&=-\frac{1}{2}\frac{f(\mathbb{Q}_0)-f'(\mathbb{Q}_0)\mathbb{Q}_0}{f'(\mathbb{Q}_0)} \label{effictive lambda}
\end{align}
\end{subequations}
\end{proof}
This result is also valid for $f(T)$ gravity. A consequence of this result is the following
\begin{proposition}
Given a function $f(\mathbb{Q})$, if there is a constant $\mathbb{Q}_0$ such that $f'(\mathbb{Q}_0)=1$ and ${\mathbb{Q}_0}-f(\mathbb{Q}_0)>0$ then in a universe filled with pressureless matter, the Friedmann equation becomes $3H^2=k\rho+\Lambda_{eff}$, with $\Lambda_{eff}=\frac{{\mathbb{Q}_0}-f(\mathbb{Q}_0)}{2}$.    
\end{proposition}
\begin{proof}
We demonstrate this for connection $\Gamma^2$. The proof proceeds identically for $\Gamma^3$. First, if we choose $\mathbb{Q}=\mathbb{Q}_0$ then the connection field equation (\ref{connection 2 equation}) is trivially satisfied, regardless of $f$ and $\gamma$. From (\ref{scalar connection 2}), we express $\gamma$ in terms of $\mathbb{Q}_0$ and $H$
\begin{equation}
    \frac{2}{3}\mathbb{Q}_0+4H^2=6H\gamma+2\dot{\gamma}
\end{equation}
Substituting into the Friedmann equation (\ref{Friedmann connection 2}),we get
\begin{equation}
    \frac{3}{4}f'(\mathbb{Q}_0)(-\frac{2}{3}\mathbb{Q}_0+4H^2)+\frac{1}{2}f=k\rho
\end{equation}
This simplifies to
\begin{equation}
    \frac{f}{2}-\frac{f'(\mathbb{Q}_0)\mathbb{Q}_0}{2}+3H^2f'(\mathbb{Q}_0)=k\rho
\end{equation}
Using $f'(\mathbb{Q}_0)=1$, the equation becomes
\begin{equation}
    3H^2=k\rho+\frac{\mathbb{Q}_0-f(\mathbb{Q}_0)}{2}
\end{equation}
Thus, the function $f(\mathbb{Q})$ reproduces $\Lambda$CDM evolution with an effective cosmological constant $\Lambda_{eff}=\frac{\mathbb{Q}_0-f(\mathbb{Q}_0)}{2}$.    
\end{proof}

This matches the effective cosmological constant derived previously in proposition \ref{theorem constant scalar}. Hence, this reconstruction result can be viewed as a direct consequence of that more general property.
However, this reconstruction is only possible when the chosen connection admits a constant $\mathbb{Q}$ without constraining the Hubble parameter.
In particular, such reconstruction cannot be achieved using connection $\Gamma^1$. In that case, constant $\mathbb{Q}$ implies constant $H$, which corresponds to the de Sitter space. Therefore, the full $\Lambda$CDM behavior cannot be recovered.
This conclusion is made explicit by the Friedmann equation for $f(\mathbb{Q})$ gravity with constant $\mathbb{Q}$ under connection $\Gamma^1$, which reads:
\begin{equation}
   \frac{f(\mathbb{Q}_0)}{2}-\mathbb{Q}_0f'(\mathbb{Q}_0)=k\rho 
\end{equation}
Since the left-hand side is constant, the energy density $\rho$ must also be constant. However, energy conservation in an expanding universe implies that $\rho$ must decay over time unless it is zero.
Therefore, for connection $\Gamma^1$, constant $\mathbb{Q}$ is incompatible with $\Lambda$CDM. On the other hand, for connection $\Gamma^2$ and $\Gamma^3$, the presence of an extra degree of freedom allows one to maintain $\mathbb{Q}$ constant, without constraining $H$. 

In summary, multiple modified theories can exactly reproduce the background dynamics of $\Lambda$CDM expansion: distinguishing these theories from General Relativity requires studying perturbations such as the growth of structure, gravitational wave propagation, and anisotropies in the cosmic microwave background. Theories which coincide at the background level may yield observationally distinct signatures at the perturbative level. Moreover, these theories must also pass local tests, a topic we address in the next section.

\section{\label{chapter solar}Solar System Tests}

GR has been remarkably successful in describing gravitational phenomena within the Solar System. A cornerstone of this success is the Schwarzschild solution, which is the unique spherically symmetric and static vacuum solution to the Einstein field equations. This uniqueness is formally established by the Birkhoff theorem \cite{Birkhoff1923-BIRRAM-3}, which states that any spherically symmetric vacuum solution in GR must be the Schwarzschild solution, even if not static a priori. Therefore, any viable modification to GR must reproduce the Schwarzschild solution (or a spacetime observationally indistinguishable from it) in regimes where it has been tested with high precision.

The Schwarzschild metric is
\begin{equation}
    ds^2=-\left(1-\frac{2GM}{r}\right)dt^2+\frac{1}{1-\frac{2GM}{r}}dr^2+r^2d\Omega^2
\end{equation}
and its linearized form given by
\begin{equation}
    ds^2=-\left(1-\frac{2GM}{r}\right)dt^2+\left(1+\frac{2GM}{r}\right)dr^2+r^2d\Omega^2
\end{equation}
We can parametrize deviations from GR by introducing the Eddington parameter $\gamma$ into the linearized metric, leading to the more general form
\begin{equation}\label{Eddington Parameter}
     ds^2=-\left(1-\frac{2GM}{r}\right)dt^2+\left(1+\gamma\frac{2GM}{r}\right)dr^2+r^2d\Omega^2
\end{equation}
{This parameter $\gamma$ influences several physical phenomena, such as the bending of light and the Shapiro time delay.
Given the metric (\ref{Eddington Parameter}), a photon passing at a distance $d$ from the Sun and forming an angle $\Phi$ between the Earth–Sun line and the incoming trajectory (with $\Phi=0$ corresponding to a source directly behind the Sun) is deflected by an angle:
\begin{equation}
    \delta\theta=\frac{1+\gamma}{2}\frac{4GM}{d}\frac{1+cos\Phi}{2}
\end{equation}
Similarly, a radar signal sent across the Solar System, passing near the Sun and reflected by a planet or satellite back to Earth, experiences an additional non-Newtonian delay in its round-trip travel time, given by
\begin{equation}
    \delta t=2(1+\gamma)M\ln\left(\frac{(r_S+\bar{r}_S\cdot\bar{n})(r_e-\bar{r}_e\cdot\bar{n})}{d^2}\right)
\end{equation}
where $\bar{e}_S$ and $\bar{r}_S$ are the position vectors from the Sun to the source and from the Sun to the Earth, respectively, and $\bar{n}$ the vector from the source to Earth.}

Precise measurements of these phenomena thus offer stringent constraints on $\gamma$ and, by extension, on deviations from GR. General Relativity predicts $\gamma = 1$, a value strongly supported by current experimental data. For instance, Very Long Baseline Interferometry yields $\gamma - 1 = (-1.6 \pm 1.5) \times 10^{-4}$ \cite{lambert2009vlbi}, while time-delay measurements from the Cassini spacecraft report $\gamma - 1 = (2.1 \pm 2.3) \times 10^{-5}$ \cite{Bertotti_delay}.

Many modified gravity theories, particularly those proposed to explain the observed acceleration of the Universe without invoking dark energy, struggle to recover the correct Solar System phenomenology. A notable example is the model proposed by Carroll et al. \cite{Carroll_2004}, in which the gravitational action is given by:
\begin{equation}
    f(R)=R-\frac{\mu^4}{R}
\end{equation}
{where $\mu$ is a new parameter with the dimensions of mass.}
This theory successfully predicted an accelerated cosmological expansion. However, it was later shown by Erickcek et al. \cite{erickcek2006} that it fails to reproduce Solar System observations. Although the SdS metric is indeed a spherically symmetric vacuum solution in this theory, it is not the unique solution. To determine the physically relevant solution, one must match the vacuum exterior to a solution for the interior of a matter distribution, such as a star. Erickcek et al. showed that when this matching is performed, the theory predicts $\gamma=\frac{1}{2}$. This result for a more general $f(R)$ was analyzed at \cite{chiba2007solar}.
In this section, we perform a similar analysis for the class of modified gravity theories introduced in the previous chapter and then generalize the result for $f(T,B)$.

\subsection{Symmetry Reduction of the Connection}

By a procedure analogous to that described in section \ref{reduction41}, the most general static and spherically symmetric metric and connection compatible with the geometric postulates of teleparallel gravity theories were determined in \cite{Black_Hole}. The most general static and spherically symmetric metric, allowing for potential off-diagonal terms, is given by
\begin{equation}\label{metric spheric big}
ds^{2}=g_{tt}dt^{2}+g_{tr}dtdr+g_{rr}dr^{2}+g_{\theta\theta}d\theta^{2}+g_{\theta\theta}\sin^{2}\theta d\phi^{2}    
\end{equation}
The general metric (\ref{metric spheric big}) can be brought into the diagonal form
\begin{equation}\label{spheric metric diagonal}
ds^{2}=g_{tt}dt^{2}+g_{rr}dr^{2}+r^{2}d\Omega^{2}    
\end{equation}
via an appropriate coordinate transformation. 
It was proven in \cite{Black_Hole}, that this coordinate transformation does not affect the connections and therefore we can use the diagonal form of the metric.

In $f(\mathbb{T)}$ gravity, the most general flat, metric-compatible spherically symmetric and static connection is
\begin{equation} \label{general f(T) connection}
\begin{aligned}
&\Gamma^t_{\hspace{0.5em}\mu\nu}=\begin{pmatrix} 
0 & 0 & 0 &0\\
\frac{\partial_rg_{tt}}{2g_{tt}} & \Gamma^t_{rr} & 0 &0\\
0 & 0 & \Gamma^t_{\theta\theta} &\frac{\Gamma^r_{\theta\phi}}{\Gamma^r_{\theta\theta}}\sin\theta\Gamma^t_{\theta\theta}\\
0 & 0 & -\frac{\Gamma^r_{\theta\phi}}{\Gamma^r_{\theta\theta}}\sin\theta\Gamma^t_{\theta\theta} & \sin^2\theta\Gamma^t_{\theta\theta}
\end{pmatrix}\\
&\Gamma^r_{\hspace{0.5em}\mu\nu}=\begin{pmatrix} 
0 & 0 & 0 &0\\
-\frac{g_{tt}}{g_{rr}}\Gamma^t_{rr} & \frac{\partial_rg_{rr}}{2g_{rr}} & 0 &0\\
0 & 0 & \Gamma^r_{\theta\theta} &\Gamma^r_{\theta\phi}\\
0 & 0 & -\Gamma^r_{\theta\phi} & \sin^2\theta\Gamma^r_{\theta\theta}
\end{pmatrix}\\
&\Gamma^\theta_{\hspace{0.5em}\mu\nu}=\begin{pmatrix} 
0 & 0 & 0 &0\\
0 & 0 & \frac{1}{r} &\Gamma^\theta_{r\phi}\\
-\frac{g_{tt}}{r^2}\Gamma^t_{\theta\theta} & -\frac{g_{rr}}{r^2}\Gamma^r_{\theta\theta} & 0 &0\\
-\frac{g_{tt}}{r^2}\Gamma^t_{\phi\theta} & \frac{g_{rr}}{r^2}\Gamma^r_{\theta\phi} & 0 & -\cos\theta\sin\theta
\end{pmatrix}\\
&\Gamma^\phi_{\hspace{0.5em}\mu\nu}=\begin{pmatrix}0 & 0 & 0 &0\\
0 & 0 & -\frac{1}{\sin^2\theta}\Gamma^\theta_{r\phi} &\frac{1}{r}\\
-\frac{g_{tt}}{\sin^2\theta r^2 \Gamma^r_{\theta\theta}}\Gamma^r_{\theta\phi}\Gamma^t_{\theta\theta} & -\frac{g_{rr}}{\sin^2\theta r^2 }\Gamma^r_{\theta\phi} & 0 &\cot\theta\\
-\frac{g_{tt}}{r^2}\Gamma^t_{\theta\theta} & -\frac{g_{rr}}{r^2}\Gamma^r_{\theta\theta} & \cot\theta & 0
\end{pmatrix}
\end{aligned}
\end{equation}

The components that were not specified are given by
\begin{subequations}
\begin{align}
\Gamma^t_{rr}&=\pm\frac{r({\Gamma^r_{\theta\theta}}^2+\frac{1}{\sin^2\theta}{\Gamma^r_{\theta\phi}}^2)\partial_rg_{rr}-2g_{rr}({\Gamma^r_{\theta\theta}}^2+\frac{1}{\sin^2\theta}{\Gamma^r_{\theta\phi}}^2-r\Gamma^r_{\theta\theta}\partial_r\Gamma^r_{\theta\theta}-r\frac{1}{\sin^2\theta}\Gamma^r_{\theta\phi}\partial_r\Gamma^r_{\theta\phi})}{2r\sqrt{g_{tt}}\sqrt{{\Gamma^r_{\theta\theta}}^2+\frac{1}{\sin^2\theta}{\Gamma^r_{\theta\phi}}^2}\sqrt{r^2-g_{rr}({\Gamma^r_{\theta\theta}}^2+\frac{1}{\sin^2\theta}{\Gamma^r_{\theta\phi}}^2)}}\\
\Gamma^t_{\theta\theta}&=\pm\frac{\Gamma^r_{\theta\theta}\sqrt{r^2-g_{rr}({\Gamma^r_{\theta\theta}}^2+\frac{1}{\sin^2\theta}{\Gamma^r_{\theta\phi}}^2)}}{\sqrt{g_{tt}}\sqrt{{\Gamma^r_{\theta\theta}}^2+\frac{1}{\sin^2\theta}{\Gamma^r_{\theta\phi}}^2}}\\
\Gamma^\theta_{r\phi}&=\frac{\Gamma^r_{\theta\theta}\partial_r\Gamma^r_{\theta\phi}-\Gamma^r_{\theta\phi}\partial_r\Gamma^r_{\theta\theta}}{{\Gamma^r_{\theta\theta}}^2+\frac{1}{\sin^2\theta}{\Gamma^r_{\theta\phi}}^2}
\end{align}
\end{subequations}
The only free components of the connection are $\Gamma^r_{\hspace{0.5em}\theta\theta}$ and $\Gamma^r_{\hspace{0.5em}\theta\phi}$. Using the non-diagonal $rt$ component of the metric field equation (\ref{metric equation f(T) Einstein form}),
\begin{equation}
f_{\mathbb{T}\mathbb{T}}\partial_r\mathbb{T}S_{(rt)}^{{\hspace{1.5em}r}}=0  
\end{equation}
Since $\mathbb{T}$ being constant leads to General Relativity, the only way to get new solutions is if the equation is an identity, choosing a connection for which $S_{(rt)}^{{\hspace{1.5em}r}}=0$. It was proven in \cite{Black_Hole} that the only way to have solutions beyond GR is by setting 
\begin{equation}
    \Gamma^r_{\hspace{0.7em}\theta\theta}=\pm\frac{r}{\sqrt{g_{rr}}}\hspace{1em}\Gamma^r_{\hspace{0.7em}\theta\phi}=0
\end{equation}
{To these two possible choices correspond the connections $\Gamma^{\pm}$:}
\begin{equation}
\label{connection f(T)}
\begin{aligned}
&\Gamma^t_{\hspace{0.5em}\mu\nu}=\begin{pmatrix} 
0 & 0 & 0 &0\\
\frac{\partial_rg_{tt}}{2g_{tt}} & 0 & 0 &0\\
0 & 0 & 0 &0\\
0 & 0 & 0 & 0
\end{pmatrix}
&\Gamma^r_{\hspace{0.5em}\mu\nu}=\begin{pmatrix} 
0 & 0 & 0 &0\\
0 & \frac{\partial_rg_{rr}}{2g_{rr}} & 0 &0\\
0 & 0 & \pm \frac{r}{\sqrt{g_{rr}}} &0\\
0 & 0 & 0 & \pm \frac{r}{\sqrt{g_{rr}}}\sin^2\theta
\end{pmatrix}\\
&\Gamma^\theta_{\hspace{0.5em}\mu\nu}=\begin{pmatrix} 
0 & 0 & 0 &0\\
0 & 0 & \frac{1}{r} & 0\\
c & \mp \frac{\sqrt{g_{rr}}}{r} & 0 &0\\
0 & 0 & 0 &-\cos\theta\sin\theta
\end{pmatrix}
&\Gamma^\phi_{\hspace{0.5em}\mu\nu}=\begin{pmatrix} 
0 & 0 & 0 &0\\
0 & 0 & 0 & \frac{1}{r}\\
0 & 0 & 0 &\cot\theta\\
0 & \mp \frac{\sqrt{g_{rr}}}{r} & \cot\theta &0
\end{pmatrix}
\end{aligned}
\end{equation}

\subsection{Physical Viability of the Connection}

When these connections were first determined in \cite{Black_Hole}, one of the main motivations was to find new black hole solutions that depart from GR. However, here we would like to keep the standard SdS solution to keep our theories consistent with solar system observation. We will argue that some physical criteria rule out connections $\Gamma^\pm$ and therefore lead to the SdS metric.

It is reasonable to expect that in the Minkowski limit all effects of gravity should vanish. In particular, both the torsion scalar and torsion tensor are expected to be zero in this limit. Computing the torsion scalar for the connections $\Gamma^{\pm}$ (\ref{connection f(T)}) we get
\begin{equation}
    \mathbb{T}^\pm=\left(1\pm\sqrt{g_{rr}}\right)\frac{2}{r^{2}g_{rr}}\left(\frac{r\partial_{r}g_{tt}}{g_{tt}}+1\pm\sqrt{g_{rr}}\right)
\end{equation}
In the case of the Minkowski metric
\begin{equation}
    \mathbb{T}^{+}=\frac{8}{r^2}\hspace{2em}\mathbb{T}^{-}=0
\end{equation}
We see that for connection $\Gamma^+$, the value of $\mathbb{T}$ for the Minkowski metric does not vanish. Only $\Gamma^-$ satisfies the natural expectation that the torsion will vanish in the Minkowski limit. It is then plausible that only this second connection is physically meaningful.

This type of problem was point out by Bahamonde et al. \cite{bahamonde2022coincident} in the context of STG. Nevertheless, these authors also caution that the vanishing of the scalar in the Minkowski limit may not be strictly necessary. Since matter follows Levi-Civita geodesics, particles are sensitive only to the metric, not to the affine structure. Therefore, as long as the metric is flat, any non-trivial geometry in the background may not be detectable. However, when considering the cosmological limit, the adequacy of these connections must also be re-examined. The de Sitter metric, when expressed in FLRW coordinates, is manifestly homogeneous and isotropic. Nonetheless, a coordinate transformation reveals that it is also spherically symmetric and static. In FLRW coordinates, the de Sitter metric takes the form
\begin{equation}
    ds^2=-dT^2+\exp(2HT)[dR^2+R^2d\Omega^2]
\end{equation}
while in static Schwarzschild-like coordinates it reads
\begin{equation}\label{metric de sitter static coordinates}
    ds^2=-(1-H^2r^2)dt^2+\frac{1}{1-H^2r^2}dr^2+r^2d\Omega^2
\end{equation}
These two forms are related via the coordinate transformation
\begin{subequations}
\label{coordinate transformation De Sitter}
\begin{align}
    r&=\exp(HT)R  \\
    t&=T-\frac{1}{2H}\log(1-H^2r^2)
\end{align}
\end{subequations}

Since de Sitter metric possesses both symmetries (homogeneity/isotropy and spherical/static), it is natural to expect that the homogeneous and isotropic connection used in cosmology should, under this coordinate transformation, match one of the spherically symmetric and static connections. In particular, the torsion scalar should remain invariant across these representations of the same space. From the cosmological analysis, we found that $\mathbb{T}=-6H^2$, which correctly vanishes in the Minkowski limit. So the connection $\Gamma^+$ which did not have the correct Minkowski limit also does not have the correct de Sitter limit. One concrete manifestation of this inconsistency is that the term $\sqrt{-\mathbb{T}}$ that comes from cosmological reconstruction is not even defined for connection $\Gamma^+$ as the torsion scalar background value is positive. 
While $\Gamma^-$ avoids this issue in the Minkowski limit, it too becomes problematic in the de Sitter regime. If we compute the torsion scalar for $\Gamma^-$, in the de Sitter metric as a series in $H^2$, we find
\begin{equation}
    \mathbb{T}=\frac{5}{2}H^4r^2+\mathcal{O}(H^6r^4)
\end{equation}
which is incompatible with $\mathbb{T}=-6H^2$.

We will formalize this concept in the following definition: a spherically symmetric and static connection is said to be de Sitter compatible if, when the metric is the de Sitter metric, the connection can be obtained from a homogeneous and isotropic connection via the coordinate transformation (\ref{coordinate transformation De Sitter}).
{Instead of merely requiring that the metric and connection are both spherically symmetric and static, we also impose that, if the metric possesses additional symmetries—as is the case for the de Sitter metric—the connection must share these symmetries. Ensuring compatibility with the de Sitter metric is, of course, trivial in General Relativity, since the Levi-Civita connection is fully determined by the metric and therefore automatically inherits all its symmetries.
In what follows, we will examine the existence of such connections in teleparallel theories and use them to obtain the desired results.} 

\subsection{Solar Systems Test in Teleparallelism}

The cosmological connection $\Gamma^0$ for de Sitter spacetime, expressed in FLRW coordinates, is given by
\begin{equation} \label{Connection f(T) k=0 de sitter}
\begin{aligned}
&\Gamma^t_{\hspace{0.5em}\mu\nu}=\bold{0}^{(4\times4)}\hspace{7em}
&\Gamma^r_{\hspace{0.5em}\mu\nu}=\begin{pmatrix} 
0 & H & 0 &0\\
0 & 0 & 0 &0\\
0 & 0 & -R &0\\
0 & 0 & 0 &-R \sin^2\theta
\end{pmatrix}\\
&\Gamma^\theta_{\hspace{0.5em}\mu\nu}=\begin{pmatrix}
0 & 0 & H &0\\
0 & 0 & \frac{1}{R} &0\\
0 &  \frac{1}{R}&0 &0\\
0 & 0 & 0 &-\cos\theta\sin\theta
\end{pmatrix}
&\Gamma^\phi_{\hspace{0.5em}\mu\nu}=\begin{pmatrix}
0 & 0 & 0 &H\\
0 & 0 & 0 &\frac{1}{R}\\
0 & 0 & 0 &\cot\theta\\
0 & \frac{1}{R} & \cot\theta &0
\end{pmatrix}  
\end{aligned}
\end{equation}
where $H$ is constant.
We now apply the coordinate transformation (\ref{coordinate transformation De Sitter}) to this connection.
Connections transform non-homogeneously following the transformation rule
\begin{equation}
\tilde{\Gamma}_{\mu\nu}^{\alpha}=\partial_{\beta}x^{\alpha}\partial_{\mu}x^{\rho}\partial_{\nu}x^{\sigma}\Gamma_{\rho\sigma}^{\beta}+\partial_{\lambda}x^{\alpha}\partial_{\mu}\partial_{\nu}x^{\lambda}
\end{equation}
Applying this to the connection (\ref{Connection f(T) k=0 de sitter}), we obtain the transformed connection in Schwarzschild coordinates
{\small
\begin{equation} \label{Connection f(T) Transformed}
\begin{aligned}
&\Gamma^t_{\hspace{0.5em}\mu\nu}=\begin{pmatrix} 
0 & 0 & 0 &0\\
\frac{-H^2r}{1-H^2r^2} & \frac{-H}{(1-H^2r^2)^2} & 0 &0\\
0 & 0 & \frac{-Hr^2}{1-H^2r^2} &0\\
0 & 0 & 0 &\frac{-Hr^2}{1-H^2r^2} \sin^2\theta
\end{pmatrix}\\
&\Gamma^r_{\hspace{0.5em}\mu\nu}=\begin{pmatrix} 
0 & 0 & 0 &0\\
-H & \frac{H^2r}{1-H^2r^2} & 0 &0\\
0 & 0 & -r &0\\
0 & 0 & 0 &-r \sin^2\theta
\end{pmatrix}\\
&\Gamma^\theta_{\hspace{0.5em}\mu\nu}=\begin{pmatrix}
0 & 0 & 0 &0\\
0 & 0 & \frac{1}{r} &0\\
-H &  \frac{1}{r}\frac{1}{1-H^2r^2}&0 &0\\
0 & 0 & 0 &-\cos\theta\sin\theta
\end{pmatrix}\\
&\Gamma^\phi_{\hspace{0.5em}\mu\nu}=\begin{pmatrix}
0 & 0 & 0 &0\\
0 & 0 & 0 &\frac{1}{r}\\
 &  \frac{1}{r}\frac{1}{1-H^2r^2}&0 &\cot\theta\\
-H & \frac{1}{r}\frac{1}{1-H^2r^2} & \cot\theta &0
\end{pmatrix}
\end{aligned}
\end{equation}}

In particular, we have $\Gamma^r_{\theta\theta}=-r$ and $\Gamma^r_{\theta\phi}=0$.
These are the two degrees of freedom that define the most general static and spherically symmetric connection in TG (\ref{general f(T) connection}). If we substitute them in connection (\ref{general f(T) connection}), we obtain a de Sitter compatible connection
\begin{equation} \label{connection f(T) compatible}
\begin{aligned}
&\Gamma^t_{\hspace{0.5em}\mu\nu}=\begin{pmatrix} 
0 & 0 & 0 &0\\
\frac{\partial_rg_{tt}}{2g_{tt}} & -\frac{\partial_{r}g_{rr}}{2\sqrt{-g_{tt}}\sqrt{g_{rr}-1}} & 0 &0\\
0 & 0 &  -r\frac{\sqrt{g_{rr}-1}}{\sqrt{-g_{tt}}} &0\\
0 & 0 & 0 & -r\sin^2 \frac{\sqrt{g_{rr}-1}}{\sqrt{-g_{tt}}}
\end{pmatrix}\\
&\Gamma^r_{\hspace{0.5em}\mu\nu}=\begin{pmatrix} 
0 & 0 & 0 &0\\
-\frac{\partial_{r}g_{rr}}{2g_{rr}}\frac{\sqrt{-g_{tt}}}{\sqrt{g_{rr}-1}} & \frac{\partial_rg_{rr}}{2g_{rr}} & 0 &0\\
0 & 0 & -r &0\\
0 & 0 &  & -r\sin^2\theta
\end{pmatrix}\\
&\Gamma^\theta_{\hspace{0.5em}\mu\nu}=\begin{pmatrix} 
0 & 0 & 0 &0\\
0 & 0 & \frac{1}{r} &0\\
-\frac{\sqrt{-g_{tt}}\sqrt{g_{rr}-1}}{r} & \frac{g_{rr}}{r} & 0 &0\\
0 & 0 & 0 & -\cos\theta\sin\theta
\end{pmatrix}\\
&\Gamma^\phi_{\hspace{0.5em}\mu\nu}=\begin{pmatrix} 
0 & 0 & 0 &0\\
0 & 0 & 0 &\frac{1}{r}\\
0 & 0 & 0 &\cot\theta\\
-\frac{\sqrt{-g_{tt}}\sqrt{g_{rr}-1}}{r} & \frac{g_{rr}}{r} & \cot\theta & 0
\end{pmatrix}
\end{aligned}
\end{equation}
For the de Sitter metric, this connection (\ref{connection f(T) compatible}) reproduces the transformed connection (\ref{Connection f(T) Transformed}). The torsion scalar for this connection (\ref{connection f(T) compatible}) reads
\begin{equation}
\mathbb{T}=-\frac{2}{r^2g_{rr}g_{tt}}[r(1-g_{rr})\partial_rg_{tt}+g_{tt}(1-g_{rr}-r\partial_rg_{rr})]
\end{equation}
which evaluates to $\mathbb{T}=-6H^2$ in the de Sitter spacetime, confirming consistency.

As discussed earlier, the connections $\Gamma^{\pm}$ are the only ones in $f(\mathbb{T})$ that allow deviations from GR. The requirement of a de Sitter compatible connection imposes that $\mathbb{T}$ must be constant, inevitably leading to the SdS solution. However, this restriction does not necessarily apply to $f(\mathbb{T},B)$ gravity. While the considerations made for de Sitter compatibility are theory-independent, the exclusivity of $\Gamma^{\pm}$ in enabling deviations from GR is a particular feature of $f(\mathbb{T})$.

In $f(\mathbb{T},B)$ gravity, the geometric and symmetry postulates are unchanged, leading to the same class of admissible connections as in $f(\mathbb{T})$.
The field equations for $f(\mathbb{T},B)$ gravity are
\begin{equation}\tag{\ref{equacao f(T,B)}}
    G_{\mu\nu}f_{T}-\frac{1}{2}g_{\mu\nu}(f-f_{T}T-f_{B}B)+2\partial_{\alpha}(f_{T}+f_{B})S_{(\mu\nu)}^{\hspace{1em}\alpha}+\nabla_{\mu}\nabla_{\nu}f_{B}-g_{\mu\nu}\square f_{B}=kT_{\mu\nu}
\end{equation}
Evaluating the $rt$ component of the field equations using the de Sitter compatible connection yields
\begin{equation}
        \partial_{r}(f_{\mathbb{T}}+f_{B})\frac{\sqrt{g_{rr}-1}}{\sqrt{-g_{tt}}r}=0\label{rt1f(T,B)}
\end{equation}
For this choice of connection, this is not an identity, and therefore imposes the constraint
\begin{equation}\label{rt2f(T,B)}
    \partial_{r}(f_{\mathbb{T}}+f_{B})=0
\end{equation}
But now, contrary to what happens in $f(\mathbb{T})$ gravity, where this equation implied $\partial_r\mathbb{T}=0$ and this for its turn implied that we cannot have solutions beyond GR, this equation can be satisfied without the scalars $\mathbb{T}$ and $B$ necessarily being constant. Instead, it simply requires a relation between them, potentially allowing other solutions beyond SdS.

We now consider the specific model
\begin{equation}
    f(\mathbb{T},B)=\mathbb{T}-H^2\frac{B}{\mathbb{T}}
\end{equation}
where $H$ is constant. For this theory, equation (\ref{rt2f(T,B)}) becomes
\begin{equation}\label{condition}
    1+H^2 \frac{B}{\mathbb{T}^2}-\frac{H^2}{\mathbb{T}}=\beta.
\end{equation}
where $\beta$ is a constant.

{To calculate the spherically symmetric static vacuum solution, we introduce a spherically symmetric mass distribution at the center as a perturbation. Since we are considering a theory that includes an effective cosmological constant, the vacuum solution will not be Minkowski space, but rather de Sitter space.}

Consider then spherically symmetric perturbations of the de Sitter metric
\begin{equation}\label{spherically-symmetric perturbed metric}
    ds^2=[-1+H^2r^2+A(r)]dt^2+[1+H^2r^2+B(r)]dt^2dr^2+r^2d\Omega^2
\end{equation}
The unperturbed value of the torsion scalar is $\mathbb{T}=-6H^2$. We define a new function
\begin{equation}\label{definition c}
    c(r)\equiv\frac{1}{6}+\frac{H^2}{\mathbb{T}}
\end{equation}
This function $c(r)<<1$ measures the departure of the torsion scalar from the unperturbed value. We expect that $c(r)\rightarrow0$ as we depart from the source of the perturbation and that it can be treated as a perturbation of the same order as $A(r)$ and $B(r)$. {We emphasize that the consistency of this procedure is ensured by employing a de Sitter–compatible connection, as it guarantees that the scalar associated with a spherically symmetric connection departs continuously from its background value.}.
Inverting the definition (\ref{definition c}), we can write the perturbed torsion to linear order in $c$, as
\begin{equation}
    \mathbb{T}=-6H^2-36H^2 c 
\end{equation}
Inserting this relation in condition (\ref{condition}) and solving for $B$ we get
\begin{equation}
    B=-6H^2[1-6(\beta-1)+6c(1-12(\beta-1))]
\end{equation}
To recover the de Sitter value $B=-18H^2$ when we turn off the perturbation, we must have $\beta=\frac{2}{3}$. Then the perturbed scalars can all be written in terms of the function $c(r)$
\begin{subequations}\label{perturbed scalars}
\begin{align}
    \mathbb{T}&=-6H^2(1+6c)\\
    B&=-18H^2(1+10c)\\
    R&=12H^2(1+12c)
\end{align}
\end{subequations}

The trace of the field equation (\ref{equacao f(T,B)}) is
\begin{equation}\label{trace equation f(T,B)}
-Rf_T-2(f-f_TT-f_BB)-3\square f_B=kT    
\end{equation}
For our model $f(\mathbb{T},B)=\mathbb{T}-H^2\frac{B}{\mathbb{T}}$, the trace equation reads
\begin{equation}\label{trace B/T}
    -R\left(1+H^2\frac{B}{\mathbb{T}^2}\right)+2H^2\frac{B}{\mathbb{T}}+3\square\frac{H^2}{\mathbb{T}}=kT 
\end{equation}
Inserting the perturbed scalars from (\ref{perturbed scalars}) in equation (\ref{trace B/T}) and linearizing in $c(r)$, we find
\begin{equation}
    \nabla^2c(r)=\frac{kT}{3}
\end{equation}
where $\nabla^2$ is the flat-space Laplacian operator. We assume that the de Sitter space is perturbed by the presence of pressureless matter, case in which $T=-\rho$. Using that $k=8\pi G$, the equation becomes
\begin{equation}\label{laplace}
    \nabla^2c(r)=-\frac{k\rho}{3}
\end{equation}
Defining the enclosed mass $m(r)=\int_0^{r}4\pi r'^2\rho(r')dr'$, we integrate to find
\begin{equation}
    \frac{dc}{dr}=-\frac{km(r)}{12\pi r^2}
\end{equation}
Outside the matter distribution $r>R_\odot$, the enclosed mass is a constant $m(r)=M$. Therefore for $r>R_\odot$ we obtain
\begin{equation}\label{solution for c}
     \frac{dc}{dr}=-\frac{kM}{12\pi r^2}\implies c(r)=\frac{kM}{12\pi r}\,.
\end{equation}
This is a vacuum solution but depends on the interior matter distribution and matches the interior solution at the boundary, reinforcing that matching conditions determine the correct solution. If instead we had assumed $\rho=0$, then the equation $\nabla^2c=0$ would admit any solution of the form $c(r)=\frac{c_1}{r}$, in particular the case $c=0$,  corresponding to the SdS solution.

Inserting the perturbed metric (\ref{spherically-symmetric perturbed metric}) and scalar quantities (\ref{perturbed scalars}) into the the $tt$ and $rr$ components of the field equations (\ref{equacao f(T,B)}), we obtain the linearized equations
\begin{subequations}
\begin{align}
  \label{erick 1}  \frac{B+B'r}{2r^2}-\nabla^2c=k\rho\\
    4c'r=B+A'r\label{erick2}
\end{align}
\end{subequations}
Using (\ref{laplace}) in (\ref{erick 1}) we get
\begin{equation}
    B+B'r=\frac{4k\rho r^2}{3}
\end{equation}
This leads to
\begin{equation}
    \partial_r(rB)=\frac{4k\rho r^2}{3}\implies B(r)=\frac{kM}{3\pi r}
\end{equation}
Using this result together with (\ref{solution for c}) in (\ref{erick2}) we get
\begin{equation}
    -\frac{kM}{3\pi r}=\frac{kM}{3\pi r}+A'r\implies A(r)=\frac{2}{3\pi}\frac{kM}{r}
\end{equation}
The solution is
\begin{equation}
    ds^2=\left(-1+H^2r^2+\frac{2kM}{3\pi r}\right)dt^2+\left(1+H^2r^2+\frac{kM}{3\pi r}\right)
\end{equation}
We observe that this solution does not reproduce the Newtonian limit $A(r)=\frac{2GM}{r}$, although this discrepancy could be circumvented by redefining the gravitational constant $k\equiv3\pi G$. 
However, the main issue lies in the fact that it predicts $\gamma=\frac{1}{2}$.

If we had naively set $\rho=0$, the trace equation would lead to
\begin{equation}
    \nabla^2c=0\implies c(r)=\frac{c_1}{r}\,;
\end{equation}
$c_1$ is an arbitrary constant that can be determined with the boundary conditions, but by making $\rho=0$ we lost that information. Now the field equations read
\begin{subequations}
\begin{align}
  B+B'r&=0\\
    -4\frac{c_1}{r}&=B+A'r\label{erick22}
\end{align}
\end{subequations}
The first equation implies that $B(r)=\frac{c_2}{r}$ and plugging this in the second we get $A=\frac{4c_1+c_2}{r}$. The Eddington parameter for this solution is 
\begin{equation}
\gamma=\frac{1}{4\frac{c_1}{c_2}+1}    
\end{equation}
In particular, for $\rho=0$, the field equations admit $c_1=0$ as a solution, which corresponds to $\gamma=1$. 

Birkhoff’s theorem is lost in this theory and there may be several spherically symmetric vacuum solutions. However, the Solar System spacetime is determined uniquely by matching the exterior vacuum solution to the interior solution. When this is done correctly, the theory predicts
\begin{equation}
    \gamma=\frac{1}{2}
\end{equation} 
in conflict with Solar System tests, which require $\gamma$ to be extremely close to unity as predicted by GR.

In analogy with the extension of Erickek et. al \cite{erickcek2006} result to a more general class of $f(R)$ theories in \cite{chiba2007solar}, we now generalize our result for $f(\mathbb{T},B)=\mathbb{T}-\alpha\frac{B}{\mathbb{T}}$ to a broader $f(T,B)$ class. 

\begin{theorem}\label{teor}
Let $f(\mathbb{T},B)$ be analytical at a constant scalar vacuum solution $(T_0,B_0)$, with negligible effective mass for the scalar mode. Assume the matter source is pressureless dust and the connection used is de Sitter compatible. Then the Eddington parameter is $\gamma=1/2$.
\end{theorem}
\begin{proof}
The metric field equations in $f(\mathbb{T},B)$ gravity are given by
\begin{equation}
    G_{\mu\nu}f_{\mathbb{T}}-\frac{1}{2}g_{\mu\nu}(f-f_{\mathbb{T}}\mathbb{T}-f_{B}B)+2\partial_{\alpha}(f_{T}+f_{B})S_{(\mu\nu)}^{\hspace{1.5em}\alpha}+\nabla_{\mu}\nabla_{\mu}f_{B}-g_{\mu\nu}\square f_{B}=kT_{\mu\nu}
\end{equation}
For a de Sitter compatible connection, the $rt$ component imposes
\begin{equation}
    \partial_{r}(f_{\mathbb{T}}+f_{B})=0
\end{equation}
This simplifies the field equations to
\begin{equation}
    G_{\mu\nu}f_{\mathbb{T}}-\frac{1}{2}g_{\mu\nu}(f-f_{\mathbb{T}}\mathbb{T}-f_{B}B)+\nabla_{\mu}\nabla_{\mu}f_{B}-g_{\mu\nu}\square f_{B}=kT_{\mu\nu}
\end{equation}
Taking the trace, we obtain
\begin{equation}
    (B-T)f_{\mathbb{T}}-2(f-f_{\mathbb{T}}\mathbb{T}-f_{B}B)-3\square f_{B}=kT
\end{equation}
In vacuum and for constant scalars
\begin{equation}
    (B_{0}-\mathbb{T}_{0})f_{\mathbb{T}}-2(f_{0}-f_{\mathbb{T}0}\mathbb{T}_{0}-f_{B0}B_{0})=0
\end{equation}
We consider a perturbed de Sitter metric of the form:
\begin{equation}
    ds^2=[-1+A(r)+\frac{R_0}{12}r^2]dt^2+[1+B(r)+\frac{R_0}{12}r^2]dt^2
\end{equation}
Expanding the scalars
\begin{equation}
\mathbb{T}=\mathbb{T}_{0}+\mathbb{T}_{1}(r) \hspace{1em} B=B_{0}+B_{1}(r)    
\end{equation}
we Taylor-expand the functions
\begin{subequations}
    \begin{align}
        f=f_{0}+f_{T0}\mathbb{T}_{1}+f_{B0}B_{1}\\
        f_{B}=f_{B0}+f_{BT0}\mathbb{T}_{1}+f_{BB0}B_{1}\\
        f_{T}=f_{T0}+f_{TT0}\mathbb{T}_{1}+f_{TB0}B_{1}\\
    \end{align}
\end{subequations}
From the condition $\partial_{r}(f_{T}+f_{B})=0$,  we find
\begin{equation}\label{constr}
    \mathbb{T}_{1}(f_{BT0}+f_{TT0})+B_{1}(f_{BB0}+f_{TB0})=0
\end{equation}
We now show that for the different ways of solving this equation, then the trace equation can always be written as a wave equation  $\nabla^2\chi+m^2\chi=\alpha\rho$.

{Case 1}: $(f_{BB0}+f_{TB0})\neq0$: Solving the constraint (\ref{constr}): 
\begin{equation}
B_{1}=-\frac{f_{BT0}+f_{TT0}}{f_{BB0}+f_{TB0}}\mathbb{T}_{1}    
\end{equation}
Substituting into the trace equation gives a wave-like equation for $\mathbb{T}_1$
\begin{equation}
\nabla^2 \mathbb{T}_{1}+m^2_{eff}\mathbb{T}_{1}=\frac{k\rho}{3\frac{f_{BT0}f_{TB0}-f_{BB0}f_{TT0}}{f_{BB0}+f_{TB0}}} 
\end{equation}
with $m^2_{eff}$ given by
\begin{equation}
    m^2_{eff}=\frac{B_{0}f_{TT0}-f_{T0}+\mathbb{T}_{0}f_{TT0}+2f_{BT0}B_{0}-\frac{f_{BT0}+f_{TT0}}{f_{BB0}+f_{TB0}}(f_{T0}+B_{0}f_{TB0}+2f_{BB0}B_{0}+f_{TB0}\mathbb{T}_{0})}{\frac{f_{BT0}f_{TB0}-f_{BB0}f_{TT0}}{f_{BB0}+f_{TB0}}}
\end{equation}

{Case 2}: $(f_{BB0}+f_{TB0})=0$: The constraint (\ref{constr}) is solved either by
Caso 2A: $(f_{BT0}+f_{TT0})=0$ or 
Caso 2B: $\mathbb{T}_{1}=0$.

In case 2A
\begin{equation}
    \nabla^2(B_{1}-\mathbb{T}_{1})+\frac{f_{TT0}(\mathbb{T}_{0}-B_{0})-f_{T0}}{3f_{BB0}}(B_{1}-\mathbb{T}_{1})=\frac{k\rho}{3f_{BB0}}
\end{equation}
and in case 2B
\begin{equation}
   \nabla^2 B_{1}-\frac{1}{3}\frac{f_{TB0}B_{0}+f_{TB0}\mathbb{T}_{0}+f_{T0}+2f_{BB0}B_{0}}{f_{BB0}}B_{1}=\frac{3k\rho}{f_{BB0}}
\end{equation}
In either case, if the mass term is negligible, then the trace equation is always
\begin{equation}\label{solv}
    \nabla^{2}f_{B}=\frac{k\rho}{3}
\end{equation}
The linearized $tt$ component of the field equations is
\begin{equation}
        G_{tt}f_{\mathbb{T}0}-\frac{1}{2}g_{tt}(f_{0}-f_{\mathbb{T}0}\mathbb{T}_{0}-f_{B0}B_{0})+\nabla^{2}f_{B}=-k\rho g_{tt}
        \end{equation}
Substituting the Einstein tensor and Equation (\ref{solv})
\begin{equation}
    \frac{1}{r^{2}}(B+rB')=\frac{2k\rho}{3f_{T0}}\implies B(r)=\frac{kM}{6\pi f_{T0}r}
\end{equation}
The linearized $rr$ component of the field equations is
\begin{equation}
    G_{rr}f_{\mathbb{T}0}-\frac{1}{4}(B_{0}-\mathbb{T}_{0})f_{\mathbb{T}0}-\frac{2}{r}\partial_rf_B=0
\end{equation}
with solution
\begin{equation}
    -\frac{1}{r^{2}}(B+A'r)-\frac{1}{f_{T0}r}\frac{kM}{6\pi r^2}=0\implies A(r)=\frac{kM}{3\pi f_{T0}r}
\end{equation}
and therefore $\gamma=\frac{1}{2}$.
\end{proof}

This demonstration also encompasses the demonstration for $f(R)$ in \cite{chiba2007solar} or \cite{sotiriou2010f}, as $f(R)=f(\mathbb{T}-B)$ corresponds to the case 2A.
{We also note that considering a scalar with a negligible mass is the physically relevant case, since $m_{eff}^2\sim H^2$ a scale that is tiny compared to that of the Solar System; the evolution of the universe is extremely slow compared to local dynamics.}

To emphasize the importance of using a de Sitter compatible connection, we will now solve the linearized equations in $f(\mathbb{T},B)$ for the $\Gamma^+$ connection which is not de Sitter compatible. For this connection, the perturbed scalars are
\begin{subequations}
\begin{align}
    \mathbb{T}&=\frac{8}{r^{2}}(1-\frac{B}{2}-\frac{rA'}{2})\\
    B&=\frac{2}{r^{2}}(4-\frac{A''}{2}r^{2}-B'r-3B-3rA')
\end{align}
\end{subequations}
and the field equations are
{\scriptsize
\begin{subequations}
\begin{align}
\frac{f}{2}-f_{\mathbb{T}}\left(\frac{\mathbb{T}}{2}+\frac{1}{r^{2}g_{rr}}-\frac{1}{r^{2}}-\frac{\partial_{r}g_{rr}}{rg_{rr}g_{rr}}\right)-\frac{f_{B}B}{2}-\frac{2}{rg_{rr}}\partial_{r}(f_{\mathbb{T}}+f_{B})(1+\sqrt{g_{rr}})+\frac{\partial_{r}^{2}f_{B}}{g_{rr}}+\frac{\partial_{r}f_{B}}{g_{rr}}\left(\frac{2}{r}-\frac{\partial_{r}g_{rr}}{2g_{rr}}\right)&=k\rho\\ 
-\frac{f}{2}+f_{\mathbb{T}}\left(\frac{\mathbb{T}}{2}+\frac{1}{r}\frac{\partial_{r}g_{tt}}{g_{tt}g_{rr}}+\frac{1}{g_{rr}r^{2}}-\frac{1}{r^{2}}\right)+\frac{f_{B}B}{2}-\frac{\partial_{r}f_{B}}{g_{rr}}\left(\frac{\partial_{r}g_{tt}}{2g_{tt}}+\frac{2}{r}\right)&=kp
\end{align}
\end{subequations}}
When we substitute the ansatz metric (\ref{spherically-symmetric perturbed metric})
and the linearized scalars in the field equations, and assuming dust-like matter
\begin{subequations}
\begin{align}
B+A'r+2H^2 r^2&=0\\
B'r+B+\frac{7}{4}H^2 r^{2}&=k\rho
\end{align}
\end{subequations}
The exterior solution is $B(r)=\frac{2GM}{r}-\frac{7H^2r^{2}}{12}$ and $A(r)=\frac{2GM}{r}-\frac{17}{24}H^2 r^2$ and the metric is
\begin{equation}\label{wrong sol}
    ds^2=[-1+\frac{2GM}{r}+\frac{7}{24}H^2 r^{2}]dt^2+[1+\frac{2GM}{r}+\frac{5}{12}H^2 r^{2}]+r^2d\Omega^2
\end{equation}
While this metric recovers the correct behavior in the $r\rightarrow0$ limit, corresponding to the Schwarzschild solution, it fails to asymptotically approach the de Sitter metric as $r\rightarrow\infty$. This discrepancy reflects the fact that the connection $\Gamma^+$ is not de Sitter compatible.

We see that the correct de Sitter limit in $f(\mathbb{T})$ imposes the SdS solution and in $f(\mathbb{T},B)$ it leads so deviation from $\gamma=1$. We presented a solution (\ref{wrong sol}) that as the correct $\gamma=1$ limit but does not match the SdS solution in the $r\rightarrow\infty$ limit.

\subsection{Solar Systems Test in Symmetric Teleparallelism}

The imposition on the connection to be static and spherically symmetric, under the geometric postulates of STG (flat and torsionless connection), leads to two solution families for the connection, as detailed in \cite{Black_Hole}.

{Solution 1}: The first class of solutions admits the following non-vanishing components for the affine connection
\begin{equation}
\label{solution1}
\begin{aligned}
&\Gamma^t_{\hspace{0.5em}\mu\nu}=\begin{pmatrix} 
c & \Gamma^\phi_{\hspace{0.5em} r\phi} & 0 &0\\
&\Gamma^\phi_{\hspace{0.5em} r\phi} & \Gamma^t_{\hspace{0.5em} rr} & 0 &0\\
0 & 0 & -\frac{1}{c} &0\\
0 & 0 & 0 &-\frac{\sin^2\theta}{c}
\end{pmatrix}
&\Gamma^r_{\hspace{0.5em}\mu\nu}=\begin{pmatrix} 
0 & 0 & 0 &0\\
0 & \Gamma^r_{\hspace{0.5em}rr} & 0 &0\\
0 & 0 & 0 &0\\
0 & 0 & 0 &0
\end{pmatrix}\\
&\Gamma^\theta_{\hspace{0.5em}\mu\nu}=\begin{pmatrix} 
0 & 0 & c &0\\
0 & 0 & \Gamma^\phi_{\hspace{0.5em}r\phi} &0\\
c & \Gamma^\phi_{\hspace{0.5em}r\phi} & 0 &0\\
0 & 0 & 0 &-\cos\theta\sin\theta
\end{pmatrix}
&\Gamma^\phi_{\hspace{0.5em}\mu\nu}=\begin{pmatrix} 
0 & 0 & 0 &c\\
0 & 0 & 0 &\Gamma^\phi_{\hspace{0.5em}r\phi}\\
0 & 0 & 0 &\cot\theta\\
c & \Gamma^\phi_{\hspace{0.5em}r\phi} & \cot\theta &0
\end{pmatrix}
\end{aligned}
\end{equation}
The undetermined components $\Gamma^t_{rr}$, $\Gamma^r_{rr}$, and $\Gamma^\phi_{r\phi}$ are constrained by the following relation:
\begin{equation}
\partial_r\Gamma^{\phi}_{\hspace{0.5em}r\phi}=c\Gamma^t_{\hspace{0.5em}rr}+\Gamma^\phi_{\hspace{0.5em}r\phi}(\Gamma^r_{\hspace{0.5em}rr}-\Gamma^\phi_{\hspace{0.5em}\rho\phi})
\end{equation}

{Solution 2}:
{\footnotesize
\begin{equation}
\label{solution2}
\begin{aligned}
&\Gamma^t_{\hspace{0.5em}\mu\nu}=\begin{pmatrix} 
k-c-c(2c-k)\Gamma^t_{\hspace{0.5em}\theta\theta} & \frac{(2c-k)\Gamma^t_{\hspace{0.5em}\theta\theta}(1+c\Gamma^t_{\hspace{0.5em}\theta\theta})}{\Gamma^r_{\hspace{0.5em}\theta\theta}} & 0 &0\\
\frac{(2c-k)\Gamma^t_{\hspace{0.5em}\theta\theta}(1+c\Gamma^t_{\hspace{0.5em}\theta\theta})}{\Gamma^r_{\hspace{0.5em}\theta\theta}} & \Gamma^t_{\hspace{0.5em} rr} & 0 &0\\
0 & 0 & \Gamma^t_{\hspace{0.5em}\theta\theta} &0\\
0 & 0 & 0 &\Gamma^t_{\hspace{0.5em}\theta\theta}\sin^2\theta
\end{pmatrix}\\
&\Gamma^r_{\hspace{0.5em}\mu\nu}=\begin{pmatrix} 
-c(2c-k)\Gamma^t_{\hspace{0.5em}\theta\theta} & c(2c-k)\Gamma^t_{\hspace{0.5em}\theta\theta}+c & 0 &0\\
c(2c-k)\Gamma^t_{\hspace{0.5em}\theta\theta}+c & \Gamma^r_{\hspace{0.5em}rr} & 0 &0\\
0 & 0 & \Gamma^r_{\hspace{0.5em}\theta\theta} &0\\
0 & 0 & 0 &\Gamma^r_{\hspace{0.5em}\theta\theta}\sin^2\theta
\end{pmatrix}\\
&\Gamma^\theta_{\hspace{0.5em}\mu\nu}=\begin{pmatrix} 
0 & 0 & c &0\\
0 & 0 & \frac{-c\Gamma^t_{\hspace{0.5em}\theta\theta}-1}{\Gamma^r_{\hspace{0.5em}\theta\theta}} &0\\
c & \frac{-c\Gamma^t_{\hspace{0.5em}\theta\theta}-1}{\Gamma^r_{\hspace{0.5em}\theta\theta}} & 0 &0\\
0 & 0 & 0 &-\cos\theta\sin\theta
\end{pmatrix}
\Gamma^\phi_{\hspace{0.5em}\mu\nu}=\begin{pmatrix} 
0 & 0 & 0 &c\\
0 & 0 & 0 &\frac{-c\Gamma^t_{\hspace{0.5em}\theta\theta}-1}{\Gamma^r_{\hspace{0.5em}\theta\theta}}\\
0 & 0 & 0 &\cot\theta\\
c & \frac{-c\Gamma^t_{\hspace{0.5em}\theta\theta}-1}{\Gamma^r_{\hspace{0.5em}\theta\theta}} & \cot\theta &0
\end{pmatrix}
\end{aligned}
\end{equation}}
These components satisfy the differential constraints
\begin{subequations}\label{constraints set 2}
    \begin{align}
        \partial_r\Gamma^t_{\hspace{0.5em}\theta\theta}&=-\frac{\Gamma^t_{\hspace{0.5em}\theta\theta}}{\Gamma^r_{\hspace{0.5em}\theta\theta}}[1+\Gamma^t_{\hspace{0.5em}\theta\theta}(3c-k+(2c-k)]-\Gamma^t_{\hspace{0.5em}rr}\Gamma^r_{\hspace{0.5em}\theta\theta}\\
        \partial_r\Gamma^r_{\hspace{0.5em}\theta\theta}&=-1-c\Gamma^t_{\hspace{0.5em}\theta\theta}(2+(2c-k)\Gamma^t_{\hspace{0.5em}\theta\theta})-\Gamma^r  _{\hspace{0.5em}rr}\Gamma^r_{\hspace{0.5em}\theta\theta}
    \end{align}
\end{subequations}

In section \ref{connection in tele}, we have shown that there is always a coordinate system where the connection vanishes. These coordinate systems were worked out in \cite{bahamonde2022coincident} for both Solution 1 and 2. They also showed that the metric tensor written in that coordinate systems has a much more complicated form, as it now encodes all the information previously contained in the connection. They conclude that nothing is gained by using the coincident gauge.

These connections can be further constrained using the metric field equations. The off-diagonal $rt$ component of the metric field equation (\ref{metric equation f(Q) Einstein form}) has the form
\begin{equation}
f_{\mathbb{Q}\mathbb{Q}}\partial_r\mathbb{Q}P^r_{\hspace{0.7em}rt}=0    
\end{equation}
Since $\mathbb{Q}$ being constant leads to General Relativity, the only way to get new solutions is if the equation is an identity, choosing a connection for which $P^r_{\hspace{0.7em}rt}=0$. It has been proven in \cite{Black_Hole} that this is impossible for Solution 1 which therefore yields only GR solutions. However, Solution 2 allows for this condition to be satisfied when additional constraints are imposed. Two distinct branches of viable connections emerge:

{Branch 1}:
\begin{equation}
    \Gamma^t_{\hspace{0.5em}\theta\theta}=\frac{k}{2c(2c-k)}\hspace{1em}
    \Gamma^t_{\hspace{0.5em}rr}=\frac{k(8c^2+2ck-k^2)}{8c^2(2c-k)^2(\Gamma^r_{\hspace{0.5em}\theta\theta})^2}\hspace{1em}c\neq0;k\neq2c
\end{equation}

{Branch 2:}
\begin{equation}
    \Gamma^t_{\hspace{0.5em}rr}=-\frac{\Gamma^t_{\hspace{0.5em}\theta\theta}}{(\Gamma^r_{\hspace{0.5em}\theta\theta})^2}\hspace{1em}c=k=0
\end{equation}

A particular solution within Branch 2, which is fully determined by the metric and is consistent with the connection field equations, is given by setting
\begin{equation}\label{conection f(Q) dependent on the metric}
    \Gamma^r_{\hspace{0.5em}\theta\theta}=\pm\frac{r}{\sqrt{g_{rr}}}
\end{equation}
These connections (\ref{conection f(Q) dependent on the metric}) were used to study static spherically symmetric solutions beyond GR in the theory $f(\mathbb{Q)}=\mathbb{Q}+\alpha\mathbb{Q}^n-2\Lambda$ in \cite{wang2025}.
Static, spherically symmetric solutions in the coincident gauge were analyzed in \cite{wang2022static,lin2021spherically}, where it was found that the exterior solutions reduce to those of GR.

For $f(\mathbb{Q})$ gravity, the cosmological connection in flat FLRW spacetime was found to be
\begin{equation}
\tag{\ref{Connection f(Q) k=0 I}}
\begin{aligned}
&\Gamma^t_{\hspace{0.5em}\mu\nu}=\bold{0}^{(4\times4)}\hspace{7em}
&\Gamma^r_{\hspace{0.5em}\mu\nu}=\begin{pmatrix}
0 & 0 & 0 &0\\
0 & 0 & 0 &0\\
0 & 0 & -R &0\\
0 & 0 & 0 &-R  \sin^2\theta
\end{pmatrix}\\
&\Gamma^\theta_{\hspace{0.5em}\mu\nu}=\begin{pmatrix}
0 & 0 & 0 &0\\
0 & 0 & \frac{1}{R} &0\\
0 &  \frac{1}{R}&0 &0\\
0 & 0 & 0 &-\cos\theta\sin\theta
\end{pmatrix}
&\Gamma^\phi_{\hspace{0.5em}\mu\nu}=\begin{pmatrix}
0 & 0 & 0 &0\\
0 & 0 & 0 &\frac{1}{R}\\
0 & 0 & 0 &\cot\theta\\
0 & \frac{1}{R} & \cot\theta &0
\end{pmatrix}
\end{aligned}
\end{equation}
Applying the coordinate transformation that relates the FLRW and static spherically symmetric representations of de Sitter spacetime, we obtain
\begin{equation} \label{Connection f(q) Transformed}
\begin{aligned}
&\Gamma^t_{\hspace{0.5em}\mu\nu}=\begin{pmatrix} 
\frac{H^3r^2}{1-H^2r^2}
 & \frac{-H^2r}{(1-H^2r^2)^2} & 0&0 \\
\frac{-H^2r}{(1-H^2r^2)^2} & \frac{-H+2H^3r^2}{(1-H^2r^2)^3} & 0 &0\\
0 & 0 &  \frac{-Hr^2}{1-H^2r^2} &0\\
0 & 0 & 0 & \frac{-Hr^2}{1-H^2r^2} \sin^2\theta
\end{pmatrix}\\
&\Gamma^r_{\hspace{0.5em}\mu\nu}=\begin{pmatrix} 
H^2r & \frac{-H}{1-H^2r^2} & 0 &0\\
\frac{-H}{1-H^2r^2} & \frac{H^2r(2-H^2r^2)}{(1-H^2r^2)^2} & 0 &0\\
0 & 0 & -r &0\\
0 & 0 & 0 &-r \sin^2\theta
\end{pmatrix}\\
&\Gamma^\theta_{\hspace{0.5em}\mu\nu}=\begin{pmatrix}
0 & 0 & -H &0\\
0 & 0 & \frac{1}{r}\frac{1}{1-H^2r^2} &0\\
-H &  \frac{1}{r}\frac{1}{1-H^2r^2}&0 &0\\
0 & 0 & 0 &-\cos\theta\sin\theta
\end{pmatrix}\\
&\Gamma^\phi_{\hspace{0.5em}\mu\nu}=\begin{pmatrix}
0 & 0 & 0 &-H\\
0 & 0 & 0 &\frac{1}{r}\frac{1}{1-H^2r^2}\\
0 &  0 &0 &\cot\theta \\
-H & \frac{1}{r}\frac{1}{1-H^2r^2} & \cot\theta &0
\end{pmatrix}
\end{aligned}
\end{equation}
This rules out the connection from solution 1 (\ref{solution1}), as in that case $\Gamma^t_{\hspace{0.5em}\theta\theta}$ is constant. If we plug $c=k=-H$, $\Gamma^r_{\hspace{0.5em}\theta\theta}=-r$ and $\Gamma^t_{\hspace{0.5em}\theta\theta}=\frac{-Hr^2}{1-H^2r^2}$ in Solution 2, the connection is equal to \ref{Connection f(q) Transformed}, confirming that Solution 2 is de Sitter compatible.

Even though Solution 1 is not de Sitter compatible, it cannot generate any solutions that deviate from General Relativity for any choice of theory. Choosing this connection results in the SdS metric independently of the form of $f(\mathbb{Q})$, and thus leaves the free parameters of the theory completely unconstrained. In contrast, Solution 2 allows for deviations from GR, but only under specific conditions. For instance, Option 1 is again ruled out, as it would require $\Gamma^t_{\hspace{0.7em}\theta\theta}$ to be constant. 
Option 2 implies a relation $\Gamma^t_{\hspace{0.7em}rr}=-\frac{\Gamma^t_{\hspace{0.7em}\theta\theta}}{(\Gamma^r_{\hspace{0.7em}\theta\theta})^2}$ which is not satisfied by the de Sitter compatible connection. Therefore, no deviation from the SdS solution arises here either.

We conclude that for $f(\mathbb{Q})$, de Sitter compatibility of the connection leads to the SdS solution, independently of the theory. In particular, for  $f(\mathbb{Q})=\mathbb{Q}+\alpha\sqrt{-\mathbb{Q}}-2\Lambda$, this solution does not depend on the parameter $\alpha$. The parameter is not constrained either by cosmology or Solar System tests.

In the case of $f(\mathbb{Q},C)$ gravity, the method previously applied to $f(\mathbb{T},B)$ gravity is still applicable. The derivation only references the connection in the non-diagonal components of the field equations. These either vanish identically for some connection choices (which were never de Sitter compatible) or enforce a constraint $\partial_r(f_\mathbb{Q}+f_C)=0$. The connection does not influence the remaining derivation in any other way. It will just depend on the Lagrangian and the metric ansatze. Therefore, the results we found in $f(\mathbb{T},B)$ gravity apply to $f(\mathbb{Q},C)$ as well.

Despite $f(\mathbb{Q})$ having more connections available, it behaves very similarly to $f(\mathbb{T})$. When choosing a connection that is de Sitter compatible, every $f(\mathbb{Q})$ and $f(\mathbb{T})$ theories lead to SdS metric. In particular, $f(\mathbb{Q})=\mathbb{Q}+\alpha\sqrt{-\mathbb{Q}}-2\Lambda$ and his torsion analogue are indistinguishable from General Relativity with a Cosmological Constant for any value of $\alpha$. Consequently, Occam's razor would suggest that these extensions are redundant unless additional evidence is found in cosmological perturbations.

All in all, the overarching outcome is that theories that use boundary terms to reconstruct a cosmological constant predict an Eddington parameter $\gamma=\frac{1}{2}$ when a de Sitter compatible connection is used and the mass term is negligible.

\section{\label{sec6}Conclusion}

{In this work, we have investigated the modified gravity theories $f(R),f(\mathbb{T}),f(\mathbb{Q})$, $f(\mathbb{T},B)$ and $f(\mathbb{Q},C)$ in two different contexts: cosmological and Solar System scales. In particular, we examined their ability to reproduce a $\Lambda$CDM-like cosmological expansion and computed their Eddington parameter $\gamma$}

In the context of cosmological reconstruction, we showed that:
\begin{itemize}
    \item For $f(R);f(\mathbb{T});f(\mathbb{Q})$ the reconstruction of the $\Lambda$CDM model requires an explicit cosmological constant. However, $f(\mathbb{T})$ and $f(\mathbb{Q})$ allow for additional terms ($\sqrt{\mathbb{-T}}$ and $\sqrt{\mathbb{-Q}}$), which do not affect the background evolution, as already proven in \cite{LCDM_Reconstruction,Gadbail_2022_Q}
    \item For the extension $f(\mathbb{T},B)$ we introduced the ansatz $f(\mathbb{T},B)=f_1(\mathbb{T})+Bf_2(\mathbb{T)}$ and and showed that infinitely many pairs $(f_1,f_2)$ can reconstruct $\Lambda$CDM. For instance, the model $f(\mathbb{T},B)=\mathbb{T}-\alpha\frac{B}{T}$ yields an effective cosmological constant $\Lambda_{eff}=3\alpha$.
    \item For connection $\Gamma_1$, $f(\mathbb{Q},C)$ gives the same results as $f(\mathbb{T},B)$. In fact, $f(\mathbb{T},B)$ is a subset of $f(\mathbb{Q},C)$.
    \item For the dynamical connections $\Gamma^2$ and $\Gamma^3$, we showed that $\Lambda$CDM reconstruction is possible even for constant $\mathbb{Q}$ thanks to the presence of extra degrees of freedom. In contrast, for $\Gamma^1$, constant $\mathbb{Q}$ implies constant $H$, and thus only the de Sitter expansion is allowed.
\end{itemize}

We then analyzed the Solar System behavior of these reconstructed theories, aiming to compute spherically symmetric and static solutions and evaluate the Eddington parameter $\gamma$. Following the methodology of \cite{erickcek2006}, we emphasized the importance of matching the exterior vacuum solution to the interior solution of the Sun and of perturbing around a de Sitter background. However, we identified a key issue: many connections presented in the literature are not consistent with a de Sitter background. To address this, we introduced the concept of a de Sitter compatible connection, defined as a static, spherically symmetric connection that can be obtained from a homogeneous and isotropic connection through the appropriate coordinate transformation of the de Sitter spacetime. We demonstrated that such connections exist in both TG and STG. With this notion in place, we proved that
\begin{itemize}
    \item in $f(\mathbb{T})$ (or $f(\mathbb{Q})$), a de Sitter compatible connection implies constant $\mathbb{T}$ (or $\mathbb{Q})$. The resulting solution is the SdS metric, and the theory predicts $\gamma=1$, independently of the function $f$.
    \item in $f(\mathbb{T},B)$ (or $f(\mathbb{Q},C)$), even with de Sitter compatible connections, $\mathbb{T}$ (or $\mathbb{Q}$) needs not be constant. Applying the matching procedure to the reconstructed theory $f(\mathbb{T},B)=\mathbb{T}-\alpha\frac{B}{T}$ (and its STG analog), we found that it yields $\gamma=\frac{1}{2}$, which is inconsistent with observations.
    \item analogously to \cite{chiba2007solar}, $\gamma=\frac{1}{2}$ holds for the $f(\mathbb{T},B)$ and $f(\mathbb{Q},C)$ theories with a negligible mass term.
\end{itemize}

{Our main result is that theories including boundary terms can produce accelerated cosmological expansion without a cosmological constant, but fail to satisfy Solar System constraints. Conversely, theories without boundary terms are consistent with Solar System observations but require a cosmological constant to drive acceleration.
This conclusion, however, is not definitive: the square-root terms to which cosmology is insensitive, can affect local dynamics and may render the effective mass term non-negligible, thereby avoiding Theorem\ref{teor}}.

Regarding the distinction among the three types of geometry and the identification of the most promising geometric mediator of gravity, our results reaffirm the ambiguity between them. In particular, when restricted to single-variable theories all three require the introduction of an explicit cosmological constant to reproduce $\Lambda$CDM evolution and when extended to include boundary terms they present the same solutions.

The fundamental difference between these theories seems to be the flexibility displayed by the connection. From this perspective, STG seems to be the most promising of the three, as it admits a wider variety of connections that could enrich its phenomenology, as we have seen in the cosmology section. However, most of those connections may be ruled out by physical criteria, as shown for spherically symmetric connections. Ultimately, the quest to understand the true geometry of gravity continues, and with it, the promise of new physics at the intersection of geometry, cosmology, and observation.

\ack{Discussions with Filipe Mena on the topic of this work are gratefully acknowledged.}

\funding{This work was financed by Portuguese funds through FCT (Funda\c c\~ao para a Ci\^encia e a Tecnologia) in the framework of the project 2022.04048.PTDC (Phi in the Sky, DOI 10.54499/2022.04048.PTDC). CJM also acknowledges FCT and POCH/FSE (EC) support through the Investigador FCT Contract 2021.01214.CEECIND/CP1658/CT0001 (DOI 10.54499/2021.01214.CEECIND/CP1658/CT0001).}

\roles{P. A. G. Monteiro: Conceptualization, Methodology, Investigation, Formal analysis, Writing – original draft; C. J. A. P. Martins: Conceptualization, Methodology, Funding acquisition, Project administration, Supervision, Writing – review \& editing.}

\data{No research data has been analyzed in this work.}

\bibliography{geometric}

@article{beltran2019trinity,
    author = "Beltr{\'a}n Jim{\'e}nez, Jose and Heisenberg, Lavinia and Koivisto, Tomi S.",
    title = "{The Geometrical Trinity of Gravity}",
    eprint = "1903.06830",
    archivePrefix = "arXiv",
    primaryClass = "hep-th",
    doi = "10.3390/universe5070173",
    journal = "Universe",
    volume = "5",
    number = "7",
    pages = "173",
    year = "2019"
}

@article{wu2024correspondence,
    author = "Wu, Cheng and Ren, Xin and Yang, Yuhang and Hu, Yu-Min and Saridakis, Emmanuel N.",
    title = "{Background-dependent and classical correspondences between $f(Q)$ and $f(T)$ gravity}",
    eprint = "2412.01104",
    archivePrefix = "arXiv",
    primaryClass = "gr-qc",
    month = "12",
    year = "2024"
}

@article{wang2025,
    author = "Wang, Wenyi and Hu, Kun and Katsuragawa, Taishi",
    title = "{Solar System tests in covariant f(Q) gravity}",
    eprint = "2412.17463",
    archivePrefix = "arXiv",
    primaryClass = "gr-qc",
    doi = "10.1103/PhysRevD.111.064038",
    journal = "Phys. Rev. D",
    volume = "111",
    number = "6",
    pages = "064038",
    year = "2025"
}

@article{LCDM_Reconstruction,
    author = "Dunsby, Peter K. S. and Elizalde, Emilo and Goswami, Rituparno and Odintsov, Sergei and Gomez, Diego Saez",
    title = "{On the LCDM Universe in f(R) gravity}",
    eprint = "1005.2205",
    archivePrefix = "arXiv",
    primaryClass = "gr-qc",
    doi = "10.1103/PhysRevD.82.023519",
    journal = "Phys. Rev. D",
    volume = "82",
    pages = "023519",
    year = "2010"
}

@MastersThesis{Thesis,
    author     =     {Monteiro, P. A. G.},
    title     =     {{Cosmological Solutions in Geometric Modified Gravity}},
    school     =     {University of Porto},
    year     =     {2025},
    }

@misc{gadbail2025_thesis,
      title={Accelerated Expansion of the Universe in Nonmetricity-based Modified Gravity}, 
      author={Gaurav N. Gadbail},
      year={2025},
      eprint={2504.08800},
      archivePrefix={arXiv},
      primaryClass={gr-qc},
      url={https://arxiv.org/abs/2504.08800}, 
}

@article{Gadbail_2023_Q_T,
    author = "Gadbail, Gaurav N. and Arora, Simran and Sahoo, P. K.",
    title = "{Reconstruction of f(Q,T) Lagrangian for various cosmological scenario}",
    eprint = "2301.08876",
    archivePrefix = "arXiv",
    primaryClass = "gr-qc",
    doi = "10.1016/j.physletb.2023.137710",
    journal = "Phys. Lett. B",
    volume = "838",
    pages = "137710",
    year = "2023"
}

@article{Gadbail_2023_Q_C,
    author = "Gadbail, Gaurav N. and De, Avik and Sahoo, P. K.",
    title = "{Cosmological reconstruction and $\Lambda $CDM universe in $f(Q,\,C)$ gravity}",
    eprint = "2312.02492",
    archivePrefix = "arXiv",
    primaryClass = "gr-qc",
    doi = "10.1140/epjc/s10052-023-12288-y",
    journal = "Eur. Phys. J. C",
    volume = "83",
    number = "12",
    pages = "1099",
    year = "2023"
}

@article{das2024nonminimal,
    author = "Das, Santanu and Mahata, Nilanjana and Ray, Priyanka",
    title = "{Cosmological implications of non-minimally coupled f(Q) gravity}",
    eprint = "2410.02318",
    archivePrefix = "arXiv",
    primaryClass = "gr-qc",
    doi = "10.1142/S0217732324500172",
    journal = "Mod. Phys. Lett. A",
    volume = "39",
    number = "12",
    pages = "2450017",
    year = "2024"
}

@article{harko2014nonminimal,
    author = "Harko, Tiberiu and Lobo, Francisco S. N. and Otalora, G. and Saridakis, Emmanuel N.",
    title = "{Nonminimal torsion-matter coupling extension of f(T) gravity}",
    eprint = "1404.6212",
    archivePrefix = "arXiv",
    primaryClass = "gr-qc",
    doi = "10.1103/PhysRevD.89.124036",
    journal = "Phys. Rev. D",
    volume = "89",
    pages = "124036",
    year = "2014"
}

@article{bertolami2007extra,
    author = "Bertolami, Orfeu and Boehmer, Christian G. and Harko, Tiberiu and Lobo, Francisco S. N.",
    title = "{Extra force in f(R) modified theories of gravity}",
    eprint = "0704.1733",
    archivePrefix = "arXiv",
    primaryClass = "gr-qc",
    doi = "10.1103/PhysRevD.75.104016",
    journal = "Phys. Rev. D",
    volume = "75",
    pages = "104016",
    year = "2007"
}

@article{Caruana_2020,
    author = "Caruana, Maria and Farrugia, Gabriel and Levi Said, Jackson",
    title = "{Cosmological bouncing solutions in $f(T,B)$ gravity}",
    eprint = "2007.09925",
    archivePrefix = "arXiv",
    primaryClass = "gr-qc",
    doi = "10.1140/S10052-020-8204-3",
    journal = "Eur. Phys. J. C",
    volume = "80",
    number = "7",
    pages = "640",
    year = "2020"
}

@article{capozziello2025extended,
    author = "Capozziello, Salvatore and Cesare, Sara and Ferrara, Carmen",
    title = "{Extended Geometric Trinity of Gravity}",
    eprint = "2503.08167",
    archivePrefix = "arXiv",
    primaryClass = "gr-qc",
    doi = "10.1140/epjc/s10052-025-14440-2",
    journal = "Eur. Phys. J. C",
    volume = "85",
    pages = "932",
    year = "2025"
}

@article{chiba2007solar,
    author = "Chiba, Takeshi and Smith, Tristan L. and Erickcek, Adrienne L.",
    title = "{Solar System constraints to general f(R) gravity}",
    eprint = "astro-ph/0611867",
    archivePrefix = "arXiv",
    doi = "10.1103/PhysRevD.75.124014",
    journal = "Phys. Rev. D",
    volume = "75",
    pages = "124014",
    year = "2007"
}

@article{erickcek2006,
    author = "Erickcek, Adrienne L. and Smith, Tristan L. and Kamionkowski, Marc",
    title = "{Solar System tests do rule out 1/R gravity}",
    eprint = "astro-ph/0610483",
    archivePrefix = "arXiv",
    doi = "10.1103/PhysRevD.74.121501",
    journal = "Phys. Rev. D",
    volume = "74",
    pages = "121501",
    year = "2006"
}

@article{unzicker2005translation,
    author = "Unzicker, Alexander and Case, Timothy",
    title = "{Translation of Einstein's attempt of a unified field theory with teleparallelism}",
    eprint = "physics/0503046",
    archivePrefix = "arXiv",
    month = "3",
    year = "2005"
}

@article{nester1998symmetric,
    author = "Nester, James M. and Yo, Hwei-Jang",
    title = "{Symmetric teleparallel general relativity}",
    eprint = "gr-qc/9809049",
    archivePrefix = "arXiv",
    reportNumber = "NCU-CCS-980904",
    journal = "Chin. J. Phys.",
    volume = "37",
    pages = "113",
    year = "1999"
}

@article{bahamonde2023teleparallel,
  title={Teleparallel gravity: from theory to cosmology},
  author={Bahamonde, Sebastian and Dialektopoulos, Konstantinos F and Escamilla-Rivera, Celia and Farrugia, Gabriel and Gakis, Viktor and Hendry, Martin and Hohmann, Manuel and Said, Jackson Levi and Mifsud, Jurgen and Di Valentino, Eleonora},
  journal={Reports on Progress in Physics},
  volume={86},
  number={2},
  pages={026901},
  year={2023},
  publisher={IOP Publishing},
month=feb
}

@article{mulder2024underdetermination,
    author = "Mulder, Ruward and Read, James",
    title = "{Is spacetime curved? Assessing the underdetermination of general relativity and teleparallel gravity}",
    eprint = "2505.04632",
    archivePrefix = "arXiv",
    primaryClass = "physics.hist-ph",
    doi = "10.1007/s11229-024-04773-y",
    journal = "Synthese",
    volume = "204",
    number = "4",
    pages = "126",
    year = "2024"
}

@article{KNOX,
title = {Newton–Cartan theory and teleparallel gravity: The force of a formulation},
journal = {Studies in History and Philosophy of Science Part B: Studies in History and Philosophy of Modern Physics},
volume = {42},
number = {4},
pages = {264-275},
year = {2011},
issn = {1355-2198},
doi = {https://doi.org/10.1016/j.shpsb.2011.09.003},
url = {https://www.sciencedirect.com/science/article/pii/S1355219811000554},
author = {Eleanor Knox},
month=nov
}

@article{mancini2025equivalent,
    author = "Mancini, Christian and Tino, Guglielmo Maria and Capozziello, Salvatore",
    title = "{Equivalent Gravities and Equivalence Principle: Foundations and Experimental Implications}",
    eprint = "2501.06487",
    archivePrefix = "arXiv",
    primaryClass = "gr-qc",
    doi = "10.1007/s10701-025-00882-x",
    journal = "Found. Phys.",
    volume = "55",
    number = "5",
    pages = "69",
    year = "2025"
}

@article{capozziello2024comparing,
    author = "Capozziello, Salvatore and Shokri, Mehdi",
    title = "{Comparing inflationary models in extended Metric-Affine theories of gravity}",
    eprint = "2408.17415",
    archivePrefix = "arXiv",
    primaryClass = "gr-qc",
    doi = "10.1016/j.dark.2024.101698",
    journal = "Phys. Dark Univ.",
    volume = "46",
    pages = "101698",
    year = "2024"
}

@article{Buchdahl,
    author = "Buchdahl, H. A.",
    title = "{Non-Linear Lagrangians and Cosmological Theory}",
    doi = "10.1093/mnras/150.1.1",
    journal = "Mon. Not. Roy. Astron. Soc.",
    volume = "150",
    number = "1",
    pages = "1--8",
    year = "1970"
}

@article{Starobinsky1980,
    author = "Starobinsky, Alexei A.",
    editor = "Khalatnikov, I. M. and Mineev, V. P.",
    title = "{A New Type of Isotropic Cosmological Models Without Singularity}",
    doi = "10.1016/0370-2693(80)90670-X",
    journal = "Phys. Lett. B",
    volume = "91",
    pages = "99--102",
    year = "1980"
}

@article{CAPOZZIELLO_2002,
    author = "Capozziello, Salvatore",
    title = "{Curvature quintessence}",
    eprint = "gr-qc/0201033",
    archivePrefix = "arXiv",
    doi = "10.1142/S0218271802002025",
    journal = "Int. J. Mod. Phys. D",
    volume = "11",
    pages = "483--492",
    year = "2002"
}

@article{Carroll_2004,
    author = "Carroll, Sean M. and Duvvuri, Vikram and Trodden, Mark and Turner, Michael S.",
    title = "{Is Cosmic Speed-Up Due to New Gravitational Physics?}",
    eprint = "astro-ph/0306438",
    archivePrefix = "arXiv",
    reportNumber = "FERMILAB-PUB-03-263-A, SU-GP-03-6-2",
    doi = "10.1103/PhysRevD.70.043528",
    journal = "Phys. Rev. D",
    volume = "70",
    pages = "043528",
    year = "2004"
}

@article{Nojiri_2003,
    author = "Nojiri, Shin'ichi and Odintsov, Sergei D.",
    title = "{Modified gravity with negative and positive powers of the curvature: Unification of the inflation and of the cosmic acceleration}",
    eprint = "hep-th/0307288",
    archivePrefix = "arXiv",
    doi = "10.1103/PhysRevD.68.123512",
    journal = "Phys. Rev. D",
    volume = "68",
    pages = "123512",
    year = "2003"
}

@article{heisenberg2024review,
    author = "Heisenberg, Lavinia",
    title = "{Review on f(Q) gravity}",
    eprint = "2309.15958",
    archivePrefix = "arXiv",
    primaryClass = "gr-qc",
    doi = "10.1016/j.physrep.2024.02.001",
    journal = "Phys. Rept.",
    volume = "1066",
    pages = "1--78",
    year = "2024"
}

@article{capozziello2023role,
    author = "Capozziello, Salvatore and De Falco, Vittorio and Ferrara, Carmen",
    title = "{The role of the boundary term in f(Q,~B) symmetric teleparallel gravity}",
    eprint = "2307.13280",
    archivePrefix = "arXiv",
    primaryClass = "gr-qc",
    doi = "10.1140/epjc/s10052-023-12072-y",
    journal = "Eur. Phys. J. C",
    volume = "83",
    number = "10",
    pages = "915",
    year = "2023"
}

@article{jimenez2018coincident,
    author = "Beltr{\'a}n Jim{\'e}nez, Jose and Heisenberg, Lavinia and Koivisto, Tomi",
    title = "{Coincident General Relativity}",
    eprint = "1710.03116",
    archivePrefix = "arXiv",
    primaryClass = "gr-qc",
    reportNumber = "NORDITA-2017-100, IFT-UAM/CSIC-17-093, ITS-ETH-2017-10",
    doi = "10.1103/PhysRevD.98.044048",
    journal = "Phys. Rev. D",
    volume = "98",
    number = "4",
    pages = "044048",
    year = "2018"
}

@article{zhao2022covariant,
    author = "Zhao, Dehao",
    title = "{Covariant formulation of f(Q) theory}",
    eprint = "2104.02483",
    archivePrefix = "arXiv",
    primaryClass = "gr-qc",
    reportNumber = "USTC-ICTS/PCFT-21-16",
    doi = "10.1140/epjc/s10052-022-10266-4",
    journal = "Eur. Phys. J. C",
    volume = "82",
    number = "4",
    pages = "303",
    year = "2022"
}

@article{krssak2016covariant,
    author = "Kr{\v{s}}{\v{s}}{\'a}k, Martin and Saridakis, Emmanuel N.",
    title = "{The covariant formulation of f(T) gravity}",
    eprint = "1510.08432",
    archivePrefix = "arXiv",
    primaryClass = "gr-qc",
    doi = "10.1088/0264-9381/33/11/115009",
    journal = "Class. Quant. Grav.",
    volume = "33",
    number = "11",
    pages = "115009",
    year = "2016"
}

@article{Heisenberg_cosmology,
    author = "D'Ambrosio, Fabio and Heisenberg, Lavinia and Kuhn, Simon",
    title = "{Revisiting cosmologies in teleparallelism}",
    eprint = "2109.04209",
    archivePrefix = "arXiv",
    primaryClass = "gr-qc",
    doi = "10.1088/1361-6382/ac3f99",
    journal = "Class. Quant. Grav.",
    volume = "39",
    number = "2",
    pages = "025013",
    year = "2022"
}

@article{Black_Hole,
    author = "D'Ambrosio, Fabio and Fell, Shaun D. B. and Heisenberg, Lavinia and Kuhn, Simon",
    title = "{Black holes in f(Q) gravity}",
    eprint = "2109.03174",
    archivePrefix = "arXiv",
    primaryClass = "gr-qc",
    doi = "10.1103/PhysRevD.105.024042",
    journal = "Phys. Rev. D",
    volume = "105",
    number = "2",
    pages = "024042",
    year = "2022"
}

@article{ayuso_degrees_of_freedom,
    author = "Ayuso, Ismael and Bouhmadi-L{\'o}pez, Mariam and Chen, Che-Yu and Chew, Xiao Yan and Dialektopoulos, Konstantinos and Ong, Yen Chin",
    title = "{Insights in $f(Q)$ cosmology: the relevance of the connection}",
    eprint = "2506.03506",
    archivePrefix = "arXiv",
    primaryClass = "gr-qc",
    reportNumber = "RIKEN-iTHEMS-Report-25",
    month = "6",
    year = "2025"
}

@article{dimakis_degrees_of_freedom,
    author = "Dimakis, N. and Paliathanasis, A. and Roumeliotis, M. and Christodoulakis, T.",
    title = "{FLRW solutions in f(Q) theory: The effect of using different connections}",
    eprint = "2205.04680",
    archivePrefix = "arXiv",
    primaryClass = "gr-qc",
    doi = "10.1103/PhysRevD.106.043509",
    journal = "Phys. Rev. D",
    volume = "106",
    number = "4",
    pages = "043509",
    year = "2022"
}

@article{guzman2024exploring,
    author = {Guzm{\'a}n, Mar{\'\i}a-Jos{\'e} and J{\"a}rv, Laur and Pati, Laxmipriya},
    title = "{Exploring the stability of f(Q) cosmology near general relativity limit with different connections}",
    eprint = "2406.11621",
    archivePrefix = "arXiv",
    primaryClass = "gr-qc",
    doi = "10.1103/PhysRevD.110.124013",
    journal = "Phys. Rev. D",
    volume = "110",
    number = "12",
    pages = "124013",
    year = "2024"
}

@article{yang_degress_of_freedom,
    author = "Yang, Yuhang and Ren, Xin and Wang, Bo and Cai, Yi-Fu and Saridakis, Emmanuel N.",
    title = "{Data reconstruction of the dynamical connection function in f(Q) cosmology}",
    eprint = "2404.12140",
    archivePrefix = "arXiv",
    primaryClass = "astro-ph.CO",
    doi = "10.1093/mnras/stae1905",
    journal = "Mon. Not. Roy. Astron. Soc.",
    volume = "533",
    number = "2",
    pages = "2232--2241",
    year = "2024"
}

@article{shi2023degrees,
    author = "Shi, Jiaming",
    title = "{Cosmological constraints in covariant f(Q) gravity with different connections}",
    eprint = "2307.08103",
    archivePrefix = "arXiv",
    primaryClass = "gr-qc",
    doi = "10.1140/epjc/s10052-023-12139-w",
    journal = "Eur. Phys. J. C",
    volume = "83",
    number = "10",
    pages = "951",
    year = "2023"
}

@article{Nojiri_Odintsov2009,
    author = "Nojiri, Shin'ichi and Odintsov, Sergei D. and Saez-Gomez, Diego",
    title = "{Cosmological reconstruction of realistic modified F(R) gravities}",
    eprint = "0908.1269",
    archivePrefix = "arXiv",
    primaryClass = "hep-th",
    doi = "10.1016/j.physletb.2009.09.045",
    journal = "Phys. Lett. B",
    volume = "681",
    pages = "74--80",
    year = "2009"
}

@article{Scalar_Field,
    author = "Elizalde, Emilio and Nojiri, Shin'ichi and Odintsov, Sergei D. and Saez-Gomez, Diego and Faraoni, Valerio",
    title = "{Reconstructing the universe history, from inflation to acceleration, with phantom and canonical scalar fields}",
    eprint = "0803.1311",
    archivePrefix = "arXiv",
    primaryClass = "hep-th",
    doi = "10.1103/PhysRevD.77.106005",
    journal = "Phys. Rev. D",
    volume = "77",
    pages = "106005",
    year = "2008"
}

@article{Gauss_Bonnet,
    author = "Cognola, Guido and Elizalde, Emilio and Nojiri, Shin'ichi and Odintsov, Sergei and Zerbini, Sergio",
    title = "{String-inspired Gauss-Bonnet gravity reconstructed from the universe expansion history and yielding the transition from matter dominance to dark energy}",
    eprint = "hep-th/0611198",
    archivePrefix = "arXiv",
    doi = "10.1103/PhysRevD.75.086002",
    journal = "Phys. Rev. D",
    volume = "75",
    pages = "086002",
    year = "2007"
}

@article{nojiri2024,
    author = "Nojiri, Shin'ichi and Odintsov, S. D.",
    title = "{Well-defined f(Q) gravity, reconstruction of FLRW spacetime and unification of inflation with dark energy epoch}",
    eprint = "2404.18427",
    archivePrefix = "arXiv",
    primaryClass = "gr-qc",
    reportNumber = "KEK-Cosmo-0344, KEK-TH-2624",
    doi = "10.1016/j.dark.2024.101538",
    journal = "Phys. Dark Univ.",
    volume = "45",
    pages = "101538",
    year = "2024"
}

@article{Gadbail_2022_Q,
    author = "Gadbail, Gaurav N. and Mandal, Sanjay and Sahoo, P. K.",
    title = "{Reconstruction of {\ensuremath{\Lambda}}CDM universe in f(Q) gravity}",
    eprint = "2210.09237",
    archivePrefix = "arXiv",
    primaryClass = "gr-qc",
    doi = "10.1016/j.physletb.2022.137509",
    journal = "Phys. Lett. B",
    volume = "835",
    pages = "137509",
    year = "2022"
}

@article{esposito2022,
    author = "Esposito, Fabrizio and Carloni, Sante and Cianci, Roberto and Vignolo, Stefano",
    title = "{Reconstructing isotropic and anisotropic f(Q) cosmologies}",
    eprint = "2107.14522",
    archivePrefix = "arXiv",
    primaryClass = "gr-qc",
    doi = "10.1103/PhysRevD.105.084061",
    journal = "Phys. Rev. D",
    volume = "105",
    number = "8",
    pages = "084061",
    year = "2022"
}

@book{Birkhoff1923-BIRRAM-3,
	address = {Cambridge},
	author = {George David Birkhoff},
	editor = {Rudolph Ernest Langer},
	publisher = {Harvard University Press},
	title = {Relativity and Modern Physics},
	year = {1923}
}

@article{Bertotti_delay,
    author = "Bertotti, B. and Iess, L. and Tortora, P.",
    title = "{A test of general relativity using radio links with the Cassini spacecraft}",
    doi = "10.1038/nature01997",
    journal = "Nature",
    volume = "425",
    pages = "374--376",
    year = "2003"
}

@article{lambert2009vlbi,
    author = "Lambert, S. B. and Le Poncin-Lafitte, C.",
    title = "{Determination of the relativistic parameter gamma using very long baseline interferometry}",
    eprint = "0903.1615",
    archivePrefix = "arXiv",
    primaryClass = "gr-qc",
    doi = "10.1051/0004-6361/200911714",
    journal = "Astron. Astrophys.",
    volume = "499",
    pages = "331",
    year = "2009"
}

@article{bahamonde2022coincident,
    author = {Bahamonde, Sebastian and J{\"a}rv, Laur},
    title = "{Coincident gauge for static spherical field configurations in symmetric teleparallel gravity}",
    eprint = "2208.01872",
    archivePrefix = "arXiv",
    primaryClass = "gr-qc",
    doi = "10.1140/epjc/s10052-022-10922-9",
    journal = "Eur. Phys. J. C",
    volume = "82",
    number = "10",
    pages = "963",
    year = "2022"
}

@article{sotiriou2010f,
    author = "Sotiriou, Thomas P. and Faraoni, Valerio",
    title = "{f(R) Theories Of Gravity}",
    eprint = "0805.1726",
    archivePrefix = "arXiv",
    primaryClass = "gr-qc",
    doi = "10.1103/RevModPhys.82.451",
    journal = "Rev. Mod. Phys.",
    volume = "82",
    pages = "451--497",
    year = "2010"
}

@article{lin2021spherically,
    author = "Lin, Rui-Hui and Zhai, Xiang-Hua",
    title = "{Spherically symmetric configuration in $f(Q)$ gravity}",
    eprint = "2105.01484",
    archivePrefix = "arXiv",
    primaryClass = "gr-qc",
    doi = "10.1103/PhysRevD.103.124001",
    journal = "Phys. Rev. D",
    volume = "103",
    number = "12",
    pages = "124001",
    year = "2021",
    note = "[Erratum: Phys.Rev.D 106, 069902 (2022)]"
}

@article{wang2022static,
    author = "Wang, Wenyi and Chen, Hua and Katsuragawa, Taishi",
    title = "{Static and spherically symmetric solutions in f(Q) gravity}",
    eprint = "2110.13565",
    archivePrefix = "arXiv",
    primaryClass = "gr-qc",
    doi = "10.1103/PhysRevD.105.024060",
    journal = "Phys. Rev. D",
    volume = "105",
    number = "2",
    pages = "024060",
    year = "2022"
}
\bibliographystyle{unsrt}
\end{document}